\documentclass[journal,onecolumn]{IEEEtran}

\usepackage{cite}
\usepackage{braket}
\usepackage{amsmath,amssymb,amsfonts}
\usepackage{algorithmic}
\usepackage{graphicx}
\usepackage{textcomp}
\def\BibTeX{{\rm B\kern-.05em{\sc i\kern-.025em b}\kern-.08em
    T\kern-.1667em\lower.7ex\hbox{E}\kern-.125emX}}

\usepackage{amsfonts}
\usepackage{amsmath}
\usepackage{amssymb}
\usepackage{array}
\usepackage{mathrsfs}
\usepackage{leftidx}
\usepackage{graphicx}
\usepackage{epstopdf}
\usepackage{dcolumn}
\usepackage{bm}
\usepackage[usenames]{color}
\usepackage{colortbl,booktabs}
\usepackage{subfigure}
\usepackage{cite}
\usepackage{multirow}
\usepackage{amsthm}

\usepackage{cases}
   
\usepackage[switch]{lineno} 
\usepackage{booktabs}
\usepackage{ulem}

\newtheorem{assumption}{Assumption}
\newtheorem{remark}{Remark}
\newtheorem{mypro}{Proposition}
\newtheorem{theorem}{Theorem}
\newtheorem{myDef}{Definition}
\newtheorem{lemma}{Lemma}
\newtheorem{corollary}{Corollary}

\begin{document}

\title{On the non-Markovian quantum control dynamics}

\author{Haijin Ding, Nina H. Amini, John E. Gough, Guofeng Zhang
\thanks{Haijin Ding was with Universit\'{e} Paris-Saclay, CNRS, CentraleSup\'{e}lec, Laboratoire des Signaux et Syst\`{e}mes (L2S), France. He is now with the Department of Applied Mathematics, the Hong Kong Polytechnic University, Hung Hom, Kowloon,  SAR, China (e-mail: dhj17@tsinghua.org.cn). }
\thanks{Nina H. Amini is with Universit\'{e} Paris-Saclay, CNRS, CentraleSup\'{e}lec, Laboratoire des Signaux et Syst\`{e}mes (L2S), France (e-mail:  nina.amini@l2s.centralesupelec.fr). }
\thanks{John E. Gough is with the Institute of Mathematics and Physics, Aberystwyth University, SY23 3BZ, Wales, United Kingdom (e-mail:  jug@aber.ac.uk). }
\thanks{Guofeng Zhang is with the Department of Applied Mathematics, The Hong Kong Polytechnic University, Hung Hom, Kowloon,  SAR, China (e-mail: guofeng.zhang@polyu.edu.hk).}
\thanks{Corresponding author: Haijin Ding.}
}

\maketitle

\begin{abstract}
In this paper, we study both open-loop control and closed-loop measurement feedback control of non-Markovian quantum dynamics arising from the interaction between a quantum system and its environment. We use the widely studied cavity quantum electrodynamics (cavity-QED) system as an example, where an atom interacts with the environment composed of a collection of oscillators.  
In this scenario, the stochastic interactions between the atom and the environment can introduce non-Markovian characteristics into the evolution of quantum states, differing from the conventional Markovian dynamics observed in open quantum systems.
As a result, the atom's decay rate to the environment varies with time and can be described by nonlinear equations. 
The solutions to these nonlinear equations can be analyzed in terms of the stability of a nonlinear system. 
Consequently, the evolution of quantum state amplitudes follows linear time-varying equations as a result of the non-Markovian quantum transient process.
Additionally, by using measurement feedback through homodyne detection of the cavity output, we can modulate the steady atomic and photonic states in the non-Markovian process.
When multiple coupled cavity-QED systems are involved, measurement-based feedback control can  influence the dynamics of high-dimensional quantum states, as well as the resulting stable and unstable subspaces.
\end{abstract}

\tableofcontents
\section{Introduction}\label{Sec:Introduction}
Quantum control has garnered much attention due to its potential applications in quantum optics~\cite{zhang2017quantum,zhang2012quantum}, quantum information processing~\cite{dong2022quantum}, quantum engineering~\cite{zagoskin2011quantum} and others~\cite{dong2010quantum}. Control methods for quantum systems can be categorized into the open-loop control and closed-loop control, similar to that in classical systems~\cite{zhang2017quantum}. In open-loop control, designed or iteratively optimized control pulses are utilized to produce required gate operations in quantum computations~\cite{khaneja2005optimal,wu2019learning}. When a control field is applied upon an atom, the atom can be excited, leading to the generation of single or multiple photons for quantum networking~\cite{li2022control}. On the other hand, closed-loop quantum control can be realized through coherent feedback or measurement feedback methods. For instance, when an atom or cavity quantum electrodynamics (cavity-QED) system is coupled to a waveguide, a coherent feedback channel can be established using photons transmitted in the waveguide~\cite{CompareCoherent}. Subsequently, the atomic dynamics and photonic states can be modulated by tuning the parameters of the coherent feedback loop~\cite{ding2023quantum,CompareCoherent,ding2023quantumNlevel}. This form of coherent feedback dynamics can be represented as a linear control system with delays determined by the length of the feedback loop~\cite{ding2023quantum,CompareCoherent}. Then the quantum state dynamics can be interpreted in terms of the stability of a linear control system with time delays~\cite{ding2023quantumNlevel,ding2023Measurement}.

In addition to quantum coherent feedback, quantum measurement feedback is another commonly employed method for feedback control in quantum systems~\cite{zhang2017quantum}. This approach involves designing feedback control based on the measurement outcomes of the quantum system~\cite{uys2018quantum,ding2023Measurement}. 
Within the framework of control theory, quantum measurement feedback control can be represented using stochastic equations due to the noise in measurement and detection apparatuses~\cite{wiseman2009quantum,GZP19,GZP20}, distinguishing it from feedback control without measurement~\cite{van2005feedback,wiseman2009quantum}.
By employing these measurement techniques and feedback control, quantum measurement feedback can influence the evolution of quantum states and facilitate the generation of desired quantum states. Consequently, quantum measurement feedback control holds significant promise for various applications in open quantum systems where the quantum states are affected by the environment~\cite{zhang2017quantum}.
Notable applications include the use of measurement-based quantum feedback control in the quantum error correction (QEC) to rectify  error bits in quantum computations~\cite{ahn2002continuous}, and preserving the coherence of a quantum state when the quantum system interacts with its environment\cite{zhang2010protecting}.

In the realm of open quantum system control, it is typically assumed that the environmental evolution timescale is considerably shorter than that of the quantum system of interest, and the decoherence rate of the quantum control system to the environment can be simplified as static~\cite{ferialdi2017exact}. This assumption allows for the modeling of the interaction between the quantum system and the environment using a master equation with a static decay rate, which is a widely employed Markovian approximation method for studying quantum dynamics in open systems. However, in numerous scenarios, this Markovian approximation proves inadequate for comprehensively analyzing the dynamics of open quantum systems~\cite{ferialdi2017exact,de2017dynamics}. For instance, in the experimental implementation utilizing nuclear magnetic resonance (NMR)~\cite{khurana2019experimental}, the NMR qubit interacts with a non-Markovian environment characterized by a randomized configuration of modulated radio-frequency fields. In such setups, the interaction between the NMR and the environment can lead to information backflow from the environment to the qubit, thus the decoherence of the qubit can be nonmonotonic, which is different from the monotonic decoherence of a qubit in the Markovian environment~\cite{breuer2002theory}. Moreover, experimental evidence of information backflow induced by non-Markovianity has also been observed in optical systems~\cite{wu2020detecting}. Due to the existence of non-Markovian noises in quantum information processing, quantum error correcting protocols based on three- or five-qubit codes~\cite{nila2026continuous}, and Petz map~\cite{biswas2025noise} can be used to enhance the performance against non-Markovian noises. Relied on these protocols, and generalized from the Markovian scenario~\cite{ahn2002continuous}, measurement feedback control based on non-Markovian quantum dynamical systems can be potentially applied in quantum error corrections~\cite{shabani2014continuous}. Given that traditional quantum control strategies within Markovian approximations are not directly applicable in these instances, it is required to explore the open-loop and closed-loop quantum control dynamics for non-Markovian open quantum systems. Additionally, it is worth noting that the traditional Morkovian scenario can be encompassed within the broader framework of non-Markovian settings as a simplified special case.

The interaction between a quantum system and its environment, characterized as non-Markovian, can be effectively represented by the integral stochastic Schr\"{o}dinger equation~\cite{diosi1997nonPLA}. This approach incorporates the influence of the environment on quantum states through a complex-valued stochastic process with an integral kennel relative with historic memory effects in the time domain~\cite{diosi1998non,diosi1997nonPLA}. Alternatively, after averaging the quantum states over the environmental noise, the non-Markovian dynamics can be described using a master equation featuring time-varying Lindblad components~\cite{DiosiPRL,yang2012nonadiabatic},  resulting from the memory effect inherent in non-Markovian dynamics~\cite{liu2019memory}. Consequently, various non-Markovian quantum control techniques have been developed including the learning~\cite{luchnikov2020machine,luchnikov2022probing} and hybrid control approaches upon the quantum system and environment~\cite{leggio2015distributed}, which are generalized from Markovian quantum control methods~\cite{delben2023control,ren2020accelerated,xie2022stochastic}.

In this paper, we utilize the commonly employed cavity-QED system depicted in Fig.~\ref{fig:NonMarkovian} as an example to investigate the dynamics of non-Marovkovian quantum control in a novel manner. We clarify that, the time-varying parameters indicating a non-Markovian process for open quantum system can be analyzed from the perspective of nonlinear process. Based on this, we illustrate the transition from non-Markovian transient dynamics to steady Markovian dynamics using the stability theory in nonlinear control. 
Subsequently, we describe the evolution of atomic and photonic states in the non-Markovian quantum dynamics with linear time varying (LTV) control equations, where external drives or the measurement feedback can influence the quantum dynamics. Lastly, we extend our analysis to encompass non-Markovian dynamics in systems where atoms in multiple coupled cavities interact with a non-Markovian environment. Then the quantum states can be categorized into the stable and unstable subspaces based on the application of measurement feedback controls.

The rest of the paper is organized as follows. Section~\ref{Sec:ModelNonMark} concentrates on the nonlinear parameter dynamics of the non-Markovian interaction between the cavity-QED system and the environment. In Section~\ref{Sec:quantumcontrolOpen}, the open-loop quantum control with the above non-Markovian interactions is analyzed from the perspective of LTV control theory. In Section~\ref{Sec:OneCavityMF}, the effects of quantum measurement feedback control on the quantum states are considered. In Section~\ref{Sec:MultiCavity}, we generalize to the circumstance of multiple coupled cavities with non-Markovian interactions with the environment, and analyze its open loop and closed loop control dynamics. Section~\ref{Sec:conclusion} concludes this paper.

\section{Non-Markovian quantum control for open systems}\label{Sec:ModelNonMark}
In this section, we introduce the quantum control modeling based on the non-Markovian interactions between a cavity-QED system and its environment. The model is made up of the linear time varying dynamics for the interactions between a multi-level atom and a cavity, and the non-Markovian decay of the quantum system to the environment governed by nonlinear dynamics.
\begin{figure}[htbp]
  \centering
  \centerline{\includegraphics[width=0.7\columnwidth]{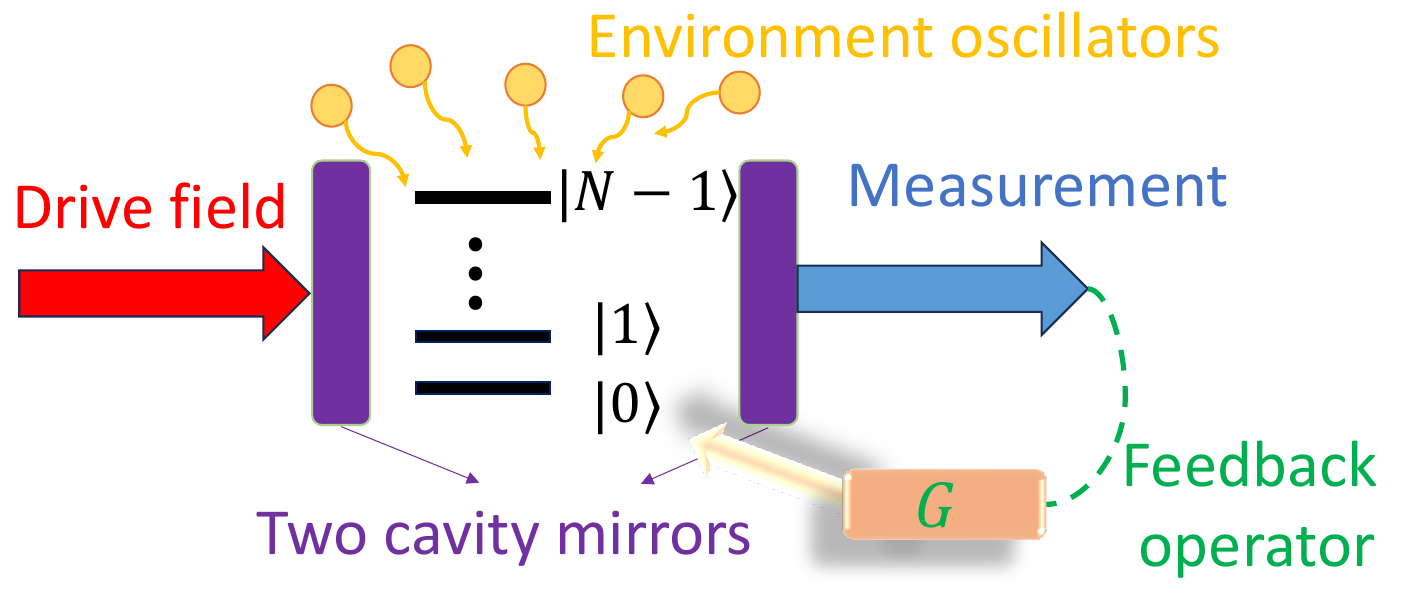}}
  \caption{Quantum control based on non-Markovian interactions between the cavity-QED system and environment.}
  \label{fig:NonMarkovian}
\end{figure}

As illustrated in Fig.~\ref{fig:NonMarkovian}, a resonant cavity is constructed between two mirrors, one multi-level atom with the energy levels $|0\rangle, |1\rangle, \cdots, |N-1\rangle$ in the format of state vectors is coupled to the cavity. The multi-level atom is simultaneously coupled to the environment modeled as a group of oscillators represented with yellow circles. The cavity can be driven by the field represented with the red arrow, and can be detected or measured via its output channel, which is represented with the blue arrow.  Using the measurement information, the feedback operator $G$ can be designed and applied upon the quantum system, thus the closed-loop control with measurement feedback can be realized.

For a simplified case without feedback, the above cavity-QED system can be modeled with the Hamiltonian 
\begin{equation} \label{con:systemHam}
H = \omega_c a^{\dag}a + \sum_{\omega} \omega b_{\omega}^{\dag}b_{\omega} +\sum_{n=0}^{N-1} \omega_n|n\rangle \langle n|  + H_I,
\end{equation}
where $\omega_c$ is the resonant frequency of the cavity with the annihilation (creation) operator $a$ ($a^{\dag}$), the second component represents a group of oscillators in the environment modeled by the annihilation (creation) operator $b_{\omega}$ ($b_{\omega}^{\dag}$) with different frequencies $\omega$, the third component is the atom Hamiltonian with $\omega_n$ represents the energy of the $n$-th level represented with the multiplication of the vector $|n\rangle$ and $\langle n|$, $H_I$ is composed with the interaction Hamiltonian between the atom and cavity, as well as the interaction Hamiltonian between the atom and environment.

For the $N$-level atom coupled to the cavity, as in Fig.~\ref{fig:NonMarkovian}, the Hamiltonian $H_I$ in Eq.~(\ref{con:systemHam}) can be equivalently represented in the interaction picture as~\cite{CompareCoherent,dong2022dynamics}
\begin{equation} \label{con:TwoNlevelatomHam}
\begin{aligned}
\bar{H}_I = &  \sum_{n=1}^{N-1} g_n\left (e^{-i\Delta_n t}\sigma^-_{n} a^{\dag} + e^{i\Delta_n t}\sigma^+_{n} a \right ) + \sum_{n=1}^{N-1}\sum_{\omega} \chi_{\omega}^{(n)} \left (  \sigma^-_{n} b_{\omega}^{\dag}+\sigma^+_{n}b_{\omega}  \right ),
\end{aligned}
\end{equation}
where we denote $\tilde{\omega}_{n} = \omega_n-\omega_{n-1}$, $\Delta_n = \tilde{\omega}_{n}- \omega_c$ is the detuning between the cavity resonant frequency and the transition frequency of two neighborhood atom energy levels, $g_n$  represents the coupling strengths between the cavity and different atom energy levels,  the lowering operator $\sigma^-_n = |n-1\rangle \langle n|$, and the rising operator $\sigma^+_n = |n\rangle \langle n-1|$. The first part on the right-hand side (RHS) represents the interaction between the atom and cavity, $\sigma^-_{n} a^{\dag}$  represents that when the atom decays from the $n$-th energy level to the $(n-1)$-th energy level, it can emit one photon into the cavity. The following Hermitian conjugate $\sigma^+_{n} a$ represents that the atom can absorb one photon from the cavity and be excited to a higher energy level.  The second part on RHS similarly represents the interaction between the multi-level atom and environment, and the coupling strength between the $n$-th energy level and the environmental oscillator with the frequency $\omega$ is $\chi_{\omega}^{(n)}$.  
$\sigma^-_{n} b_{\omega}^{\dag}$ represents that an emitted filed generated by the atom's decay from the $n$-th to $(n-1)$-th energy level can be absorbed by the environment, similar for the following Hermitian conjugate. For convenience, we denote $\tilde{g}_n = g_ne^{i\Delta_n t}$ in the following.

Different energy levels of the ladder type atom can be coupled to the environment with the operator~\cite{chen2014exact,jing2010non,mondal2020formation}
\begin{equation} \label{con:Loperator1}
\begin{aligned}
L =  \sum_{n=1}^{N-1}  L_n =\sum_{n=1}^{N-1} \sqrt{\kappa_n} |n-1\rangle \langle n|,
\end{aligned}
\end{equation}
where $\kappa_n$ represents the decay rate of the atomic state $|n\rangle$ to the environment and $L_n = \sqrt{\kappa_n} |n-1\rangle \langle n|$. More details are introduced in Appendix~\ref{sec:NonMarkovParemeters}.

The dynamics of the above quantum system can be modeled with the Schr\"{o}dinger equation $|\dot{\psi}(t)\rangle = -i H |\psi(t)\rangle$, and the influence by the environmental oscillators can be modeled as a stochastic process with Markovian or non-Markovian statistical properties, as introduced in Refs.~\cite{diosi1998non,diosi1997nonPLA,yu1999non,de2005two}. Due to Eq.~(\ref{con:Loperator1}), the Schr\"{o}dinger equation can be written as ~\cite{diosi1998non,diosi1997nonPLA,yu1999non,de2005two}
\begin{equation} \label{con:SSENonMar}
\begin{aligned}
\frac{\mathrm{d}}{\mathrm{d}t} |\psi(t)\rangle =& -i\bar{H} |\psi(t)\rangle + \left[L |\psi(t)\rangle z_t -L^{\dag}\int_0^t \alpha(t,s)\frac{\delta |\psi(t)\rangle}{\delta z_s}\mathrm{d}s\right],
\end{aligned}
\end{equation}
where $\bar{H}$ represents the interaction component between the atom and cavity in Eq.~(\ref{con:TwoNlevelatomHam}), $z_t$ is a complex stochastic process determined by the environment with $E(z_t^*z_s) = \alpha(t,s)$, $E(z_tz_s) = 0$, $\delta |\psi(t)/\delta z_s$ is a functional derivative of $|\psi(t)\rangle$ with respect to $z_s$, and
\begin{equation} \label{con:NonMarkovAlpha}
\alpha(t,s) = \frac{\gamma}{2} e^{-\gamma|t-s|-i\Omega(t-s)},
\end{equation}
where $\gamma^{-1}$ represents the environmental memory time scale and $\Omega$ represents the environmental central frequency for the modeled oscillators~\cite{diosi1998non}. Generalized from the definition of a quantum Markov process in \cite{AccardiFrigerioLewis}, the quantum dynamics governed by equation (\ref{con:SSENonMar}) is non-Markovian when $\gamma$ is finite, as it includes a memory term wherein the quantum state   $|\psi(t)\rangle$ depends on past noise $z_s$, weighted by a nonzero kernel $\alpha(t,s)$ for $s < t$~\cite{diosi1997nonPLA}.

Alternatively, the quantum dynamics in Eq.~(\ref{con:SSENonMar}) can be equivalently represented by the density matrix $\rho(t) = |\psi(t)\rangle \langle  \psi(t) |$, whose dynamics is  governed by the following master equation according to the derivation in Appendix~\ref{sec:NonMarkovMaster},
\begin{equation} \label{con:SSENonMarUpdate0}
\begin{aligned}
\dot{\rho} =&-i[\bar{H},\rho] + \sum_{n=1}^{N-1}\left[ F_n(t) + F_n^*(t)\right]   L_n  \rho L_n^{\dag} - \sum_{n=1}^{N-1}  \left[ F_n(t)  L_n^{\dag}L_n\rho + F_n^*(t) \rho L_n^{\dag}L_n\right] ,
\end{aligned}
\end{equation}
where
\begin{equation} \label{con:NonMarkovStep3Ft}
F_n(t) = \int_0^t \alpha(t,s)f_n(t,s)\mathrm{d}s,
\end{equation}
$f_n(t,s)$ is a function to be determined according to the memory kernel $\alpha(t,s)$~\cite{diosi1998non}. We denote the reduced value $\chi_n=f_n(t,t)$ when $t=s$, and $\chi_n$ can be different constants for different channels indexed by $n$~\cite{diosi1998non}. Then due to Eq.~(\ref{con:FtEquation}) in Appendix~\ref{sec:NonMarkovParemeters}, the trajectory of $F_n(t)$ can equivalently rely on $\chi_n$ rather than $f_n(t,s)$. For the $n$-th Lindbladian in Eq.~(\ref{con:SSENonMarUpdate0}), we briefly introduce the physical meanings in Table~\ref{tab:meaning}.
\begin{table}[htbp] 
    \centering 
     \caption{Physical meanings of the components in Eq.~(\ref{con:SSENonMarUpdate0})}
     \label{tab:meaning}
\begin{tabular}{cccc}
   
   \toprule
   Index & Decaying  operator & Integral kernel & Interaction rate \\
   \midrule
   $n$ & $L_n$ & $\alpha(t,s)$ & $F_n(t)$ \\
   \bottomrule
\end{tabular}
\end{table}

The reduced Markovian case is distinguished as follows.
\begin{remark} \label{RemarkMarkov}
When $\gamma \rightarrow \infty$, the interaction between the quantum system and environment is Markovian with $\alpha(t,s) = \delta(t-s)$, then~\cite{diosi1998non,AccardiFrigerioLewis}
\begin{equation} \label{con:MarkovFt}
F_n(t) = \int_0^t \delta(t-s)f_n(t,s)\mathrm{d}s =\frac{\chi_n}{2},
\end{equation}
equals a constant according to Appendix~\ref{sec:NonMarkovParemeters}.
\end{remark}

\subsection{Single multi-level atom coupled to the cavity}
Based on the non-Markovian Schr\"{o}dinger equation~(\ref{con:SSENonMar}) in a vector format, or the non-Markovian master equation~(\ref{con:SSENonMarUpdate0}) in a matrix format, the non-Markovian atomic dynamics can be evaluated by the mean-value of an arbitrary operator $\textbf{O}$ is governed by $\dot{\langle \textbf{O} \rangle} = {\rm  Tr} \left[ \dot{\rho} \textbf{O} \right] $~\cite{wiseman2009quantum}. For example, the operator mean value $\langle \sigma^+_{n} \sigma^-_{n}\rangle$ represents the population or probability that the atom is excited at the $n$-th energy level. Because of the coupling between the atom and the cavity, the dynamics of $\langle \sigma^+_{n} \sigma^-_{n}\rangle$ is affected by the exchange of energy between the atom and the cavity, which can be evaluated by the dynamics of the operator mean value $\langle \sigma^+_n a\rangle$, representing that the atom can absorb one photon from the cavity to be excited to a higher energy level, or the reverse process evaluated by $\langle \sigma^-_n a^{\dag}\rangle$. Besides, the number of photons in the cavity can be evaluated by the mean value $\langle a^{\dag} a\rangle$~\cite{meystre2021quantum}. Then we can derive the equation of the mean values of the operators as~\cite{dephaseCavity,wang2021broadband,he1994two,ruiz2014spontaneous}
\begin{subequations} \label{eq:NlevelOperatorEquation}
\begin{align}
&\frac{\mathrm{d}}{\mathrm{d} t} \langle \sigma^+_{n} \sigma^-_{n}\rangle=  -i\tilde{g}_n \langle \sigma^+_n a \rangle +i  \tilde{g}_n^* \langle \sigma^-_n a^{\dag} \rangle +i\tilde{g}_{n+1} \langle \sigma^+_{n+1} a \rangle -i  \tilde{g}_{n+1}^* \langle \sigma^-_{n+1} a^{\dag} \rangle \notag\\
& ~~~~~~~~~~~~~~~~- \left ( F_n(t) + F_n^*(t)\right) \kappa_n  
 \langle \sigma^+_{n} \sigma^-_{n}\rangle 
 +\left ( F_{n+1}(t) + F_{n+1}^*(t)\right) \kappa_{n+1}  
 \langle \sigma^+_{n+1} \sigma^-_{n+1}\rangle,\label{NlevelPopu}\\
&\frac{\mathrm{d}}{\mathrm{d} t} \langle \sigma^+_n a\rangle= - i \Delta_n \langle \sigma^+_n a\rangle -i \tilde{g}_n^*  \langle \sigma^+_{n} \sigma^-_{n} \rangle + i\tilde{g}_n^* \langle a^{\dag} a \rangle - F_n^*(t) \kappa_n    \langle \sigma^+_{n} a\rangle, \label{SPa}\\
&\frac{\mathrm{d}}{\mathrm{d} t} \langle \sigma^-_n a^{\dag}\rangle=  i \Delta_n \langle \sigma^-_n a^{\dag}\rangle +i \tilde{g}_n  \langle \sigma^+_{n} \sigma^-_{n} \rangle - i\tilde{g}_n  \langle a^{\dag} a \rangle - F_n(t) \kappa_n   \langle \sigma^-_{n} a^{\dag}\rangle, \label{SMad}\\
&\frac{\mathrm{d}}{\mathrm{d} t} \langle a^{\dag} a\rangle= i\sum_n  \left( \tilde{g}_n   \langle \sigma^+_n a\rangle - \tilde{g}_n^*  \langle \sigma^-_n a^{\dag} \rangle \right),\label{NlevelPhoton}\\
&\frac{\mathrm{d}}{\mathrm{d} t}F_n(t) = \kappa_n F_n^2(t)  - \left(\gamma + i\Omega - i\tilde{\omega}_{n}\right)F_n(t) + \frac{\gamma \chi_n}{2}, \label{NMEqF}
\end{align}
\end{subequations}
where $n =1,2,\dots,N-1$. Eq.~(\ref{NlevelPopu}) represents that the atomic state excited at the $n$-th energy level can be acquired from a lower energy level by absorbing one photon from the cavity or from a higher energy level by emitting one photon into the cavity, which are represented by the first four components on the RHS, and the following components represent the non-Markovian interaction between the atom and environment, which can make the atom decay to a lower energy level. Eq.~(\ref{SPa}) and Eq.~(\ref{SMad}) represent the interface between the atom and the cavity via the emitting and absorbing processes of a photon. This process is influenced by the detuning between the atom and the cavity (i.e., the first component of the RHS of Eq.~(\ref{SPa}) and Eq.~(\ref{SMad})), the coupling strengths $\tilde{g}_n$, and the non-Markovian decay to the environment represented by the last component of the RHS of Eq.~(\ref{SPa}) and Eq.~(\ref{SMad}), respectively. Eq.~(\ref{NlevelPhoton}) represents how the number of photons in the cavity is influenced by the coupling between the atom and cavity. The nonlinear Eq.~(\ref{NMEqF}) represents the non-Markovian decay rate of the atom to the environment, which is derived in detail in  Appendix~\ref{sec:NonMarkovParemeters}.
Additionally, the atom populations are normalized as $\langle \sigma_{00} \rangle + \sum_{n=1}^{N-1} \langle \sigma^+_{n} \sigma^-_{n}\rangle  =1$, where $\langle \sigma_{00} \rangle $ represents the population that the atom is at the ground state and $ \langle \dot{\sigma}_{00} \rangle =  \left [F_1 + F_1^*\right] \kappa_1   \langle \sigma^+_{1} \sigma^-_{1}\rangle +i\tilde{g}_1 \langle \sigma^+_1 a \rangle - i\tilde{g}_1^* \langle \sigma^-_1 a^{\dag}\rangle$.

We define the complex-valued state vector
\begin{equation} \label{con:XntDef}
\begin{aligned}
X(t) = \left [X_1(t), X_2(t),\cdots, X_{N-1}(t) , \langle a^{\dag} a\rangle \right ]^{\rm T},
\end{aligned}
\end{equation}
with $X_n(t) = \left [ \langle \sigma^+_{n} \sigma^-_{n}\rangle,\langle \sigma^+_n a\rangle, \langle \sigma^-_n a^{\dag}\rangle \right ]^{\rm T}$ and ${\rm T}$ denoting the matrix transpose. Then Eq.~(\ref{eq:NlevelOperatorEquation}) can be rewritten as
\begin{equation} \label{con:TimevariedEquation}
\begin{aligned}
\frac{\mathrm{d}X}{\mathrm{d} t}  &= \begin{bmatrix}
A_1(t) & B_1(t) & \cdots & 0 & D_1(t)\\
0 & A_2(t) & \cdots & 0 & D_2(t)\\
 \vdots &\vdots &\ddots  &\vdots  &\vdots\\
 0 & 0 & \cdots & A_{N-1}(t) & D_{N-1}(t)\\
C_1(t) &C_2(t) & \cdots  & C_{N-1}(t) &0
\end{bmatrix}X\\
&\triangleq  A(t) X,
\end{aligned}
\end{equation}
where
\begin{equation} \label{con:TVEAnt}
\begin{aligned}
A_n=&-\kappa_n {\rm{diag}}\left[(F_n + F_n^*), F_n^*,F_n\right] +\begin{bmatrix}
 0  & -i \tilde{g}_n  & i \tilde{g}_n^* \\
-i \tilde{g}_n^*& -i\Delta_n   & 0\\
i \tilde{g}_n & 0 &  i\Delta_n 
\end{bmatrix},
\end{aligned}
\end{equation}
\begin{equation} \label{con:TVEBnt}
\begin{aligned}
&B_n=\begin{bmatrix}
\kappa_{n+1} \left(F_{n+1} + F_{n+1}^*\right) & i \tilde{g}_{n+1}  & -i \tilde{g}_{n+1}^* \\
0&0 & 0\\
0 & 0 &0
\end{bmatrix},
\end{aligned}
\end{equation}
\begin{equation} \label{con:TVECnt}
\begin{aligned}
C_n&=\begin{bmatrix}
0 & i \tilde{g}_n & -i \tilde{g}_n^*
\end{bmatrix},
\end{aligned}
\end{equation}
\begin{equation} \label{con:TVDCnt}
\begin{aligned}
D_n&=\begin{bmatrix}
0 & i \tilde{g}_n^* & -i \tilde{g}_n 
\end{bmatrix}^{\rm T},
\end{aligned}
\end{equation}
with $n=1,2,\cdots,N-1$. 

\begin{remark}
By deriving the dynamics for the mean values of the operators representing atom's populations, atom-cavity interactions, and the number of photons in the cavity, the non-Markovian equations (\ref{con:SSENonMar},\ref{con:SSENonMarUpdate0}) containing integral interactions with the environment can be analyzed from the perspective of linear time-varying dynamics in Eq.~(\ref{con:TimevariedEquation}).
\end{remark}

For the basic control model of this paper in Eq.~(\ref{con:TimevariedEquation}), the parameter matrices are time-varying because the decay rates of the multi-level atom to the environment $F_n(t)$ are determined by time-varying integrations with memory effects. The target of this paper is to study the non-Markovian interactions between quantum system and environment evaluated by the nonlinear dynamics in Eq.~(\ref{NMEqF}), and how the stability of time-varying atomic and photonic dynamics in Eq.~(\ref{con:TimevariedEquation}) can be influenced by the non-Markovian process and additional control methods.

\subsection{Non-Markovian parameter dynamics} 
In this subsection, we clarify the dynamics of $F_n(t)$ from the perspective of nonlinear dynamics, and then analyze how the time-varying $F_n(t)$ can influence the quantum control performance in Eq.~(\ref{eq:NlevelOperatorEquation}) in the following Sec.~\ref{Sec:quantumcontrolOpen}.

In Eq.~(\ref{NMEqF}), we denote 
\begin{subequations} \label{eq:PQR}
\begin{numcases}{}
Q_n = -\left (\gamma +i \Omega -i\tilde{\omega}_{n} \right),\\
S_n = \frac{\gamma \chi_n}{2},
\end{numcases}
\end{subequations}
then Eq.~(\ref{NMEqF}) can be simply rewritten as
\begin{equation} \label{con:rewriteF}
\begin{aligned}
\frac{\mathrm{d}}{\mathrm{d} t}F_n(t) = \kappa_n F_n^2(t)  + Q_nF_n(t) + S_n,
\end{aligned}
\end{equation}
where $F_n(t)$ can be solved according to following different parameter settings.

\subsubsection{$\Omega = \tilde{\omega}_{n}$ and $ 4\kappa_nS_n- Q_n^2 <  0$}
When the conditions are satisfied, namely $4\kappa_nS_n  < Q_n^2$ and $2\chi_n  \kappa_n < \gamma$, which means that the coupling between the atom and the environment can be relatively small and  bounded by $\gamma$, then
\begin{equation}
\begin{aligned} \label{con:F1}
F_n=\frac{2 \sqrt{ \left (\frac{Q_n}{2\kappa_n}\right)^2 - \frac{S_n}{\kappa_n}}}{1-e^{C + \sqrt{ Q_n^2-4\kappa_nS_n}  t}} - \frac{Q_n+ \sqrt{Q_n^2-4\kappa_nS_n}}{2\kappa_n} ,
\end{aligned}
\end{equation}
where $C$ is a constant and can be determined by $F_n(0)$. For example, $e^C=1- 2 \sqrt{ Q_n^2 - 4S_n\kappa_n}\left[  Q_n + \sqrt{ Q_n^2 -4 S_n\kappa_n} \right]^{-1}$ if $F_n(0) = 0$.

\subsubsection{$\Omega = \tilde{\omega}_{n}$ and $4\kappa_nS_n- Q_n^2> 0$}
In this case, $2\chi_n \kappa_n > \gamma$, which means that the coupling between the atom and environment is relatively stronger. Then by Eq.~(\ref{con:rewriteF}),
\begin{equation}
\begin{aligned}
F_n=& - \frac{Q_n}{2\kappa_n}+\sqrt{\frac{S_n}{\kappa_n}- \left (\frac{Q_n}{2\kappa_n}\right)^2} \tan \left[\sqrt{\frac{S_n}{\kappa_n}- \left (\frac{Q_n}{2\kappa_n}\right)^2} \left (\kappa_nt +C \right)\right] ,
\end{aligned}
\end{equation}
where $C$ can be similarly determined by $F_n(0)$ and the solution is not unique.

\subsubsection{$\Omega = \tilde{\omega}_{n}$ and $ 4\kappa_nS_n- Q_n^2= 0$} \label{Sec:ParaExample}
In this case, $2\chi_n  \kappa_n = \gamma$, and
\begin{equation}
\begin{aligned} \label{con:F3}
F_n =- \frac{Q_n}{2\kappa_n} -\frac{1}{\kappa_n t + C},
\end{aligned}
\end{equation}
with $C=-2\kappa_n/Q_n$ if $F_n(0) = 0$.

We have the following theorem for the relationship between the parameter settings above and the non-Markovian dynamics.

\begin{theorem} \label{Fproperty}
When $\tilde{\omega}_{n} = \Omega $ and $4\kappa_nS_n- Q_n^2 \leq 0$, $F_n(t)$ in Eq.~(\ref{con:rewriteF}) converges to a constant  when $t\rightarrow \infty$, resulting the trajectories for the population of atomic eigenstates related to the $n$-th level, namely $\langle \sigma^+_{n} \sigma^-_{n}\rangle$, $\langle \sigma^+_n a\rangle$ and $\langle \sigma^-_n a^{\dag}\rangle$ in Eq.~(\ref{eq:NlevelOperatorEquation}),  finally converge to be Markovian when $t\rightarrow \infty$.
\end{theorem}
\begin{proof}
According to the calculations above, when $\Omega = \tilde{\omega}_{n}$ and $4\kappa_nS_n- Q_n^2 \leq  0$,
\begin{equation}
\lim_{t\rightarrow \infty} F_n(t) =- \frac{Q_n+ \sqrt{Q_n^2-4\kappa_nS_n}}{2\kappa_n}.
\end{equation}

When $4\kappa_nS_n- Q_n^2> 0$, $\lim_{t\rightarrow \infty} F_n(t)$ does not exist because $ \tan \left[\sqrt{\frac{S_n}{\kappa_n}- \left (\frac{Q_n}{2\kappa_n}\right)^2} \left (\kappa_nt +C \right)\right]$ is infinite when 
\[\sqrt{\frac{S_n}{\kappa_n}- \left (\frac{Q_n}{2\kappa_n}\right)^2} \left (\kappa_nt +C \right) = l\pi + \frac{\pi}{2},\] 
with $l =0,1,2,\cdots$.
\end{proof}

Based on this, we further generalize to the dynamics of arbitrary energy levels according to the stability of nonlinear dynamics in the following.

For the general case that $F_n(t)$ is complex when $\Omega \neq \tilde{\omega}_{n}$, we denote $F_n(t) = R_n(t) + i I_n(t)$ with $R_n(t)$ and $I_n(t)$ representing the real and imaginary part of $F_n(t)$ respectively, then we can derive the following real-valued nonlinear equations
\begin{subequations} \label{eq:FRealI}
\begin{numcases}{}
\frac{\mathrm{d}  R_n}{\mathrm{d} t} =  \kappa_n \left( R_n^2 - I_n^2 \right) - \gamma R_n + \left( \Omega -\tilde{\omega}_{n}\right)I_n + \frac{\gamma \chi_n}{2}, \label{FRealILastV} \\
\frac{\mathrm{d}I_n}{\mathrm{d} t}  = 2 \kappa_n R_nI_n -  \left( \Omega -\tilde{\omega}_{n}\right)R_n - \gamma I_n.
\end{numcases}
\end{subequations}
Denote $X_F(t) = \left [R_n(t), I_n(t)\right]^{\rm T}$ with $X_F(0) = \left [0, 0\right]^{\rm T}$. Then Eq.~(\ref{eq:FRealI}) can be rewritten as
\begin{equation} \label{con:NonlinearXF}
\begin{aligned}
\frac{\mathrm{d}}{\mathrm{d} t} X_F(t) &= \begin{bmatrix}
\kappa_n R_n - \gamma & -\kappa_n I_n \\
\kappa_n I_n & \kappa_n R_n - \gamma
\end{bmatrix} \begin{bmatrix}
R_n \\
I_n
\end{bmatrix}  + \begin{bmatrix}
0 & \left( \Omega - \tilde{\omega}_{n}\right) \\
-  \left( \Omega - \tilde{\omega}_{n}\right) & 0
\end{bmatrix} \begin{bmatrix}
R_n(t) \\
I_n(t)
\end{bmatrix}  + \begin{bmatrix}
\frac{\gamma\chi_n}{2} \\
0
\end{bmatrix}\\
&\triangleq  \textbf{f}\left(X_F(t)\right) + \textbf{f}_{\Omega} X_F(t),
\end{aligned}
\end{equation}
where $\textbf{f}\left(X_F(t)\right)$ represents the sum of the first and third terms after the first 
equal sign, and $\textbf{f}_{\Omega} = \begin{bmatrix}
0 & \left( \Omega - \tilde{\omega}_{n}\right) \\
-  \left( \Omega - \tilde{\omega}_{n}\right) & 0
\end{bmatrix}$.

Based on the contraction analysis for arbitrary initial condition of $X_F(t)$~\cite{Automatica2008contraction}, consider the differential in Eq.~(\ref{con:NonlinearXF}), then 
\begin{equation}
\begin{aligned}\label{con:differential}
\delta \dot{X}_F(t) =\left [  \frac{\partial \textbf{f}\left(X_F(t)\right)}{\partial X_F} + \textbf{f}_{\Omega} \right]\delta X_F(t),
\end{aligned}
\end{equation}
and
\begin{equation}
\begin{aligned}\label{con:differentialSquare}
&\frac{\mathrm{d}}{\mathrm{d} t}\left ( \delta X_F(t)^{\rm T} \delta X_F(t) \right) = 2\delta X_F(t)^{\rm T} \frac{\partial \textbf{f}}{\partial X_F}\delta X_F(t) +  2 \textbf{f}_{\Omega} \delta X_F(t)^{\rm T}\delta X_F(t).
\end{aligned}
\end{equation}

The stability of Eq.~(\ref{con:NonlinearXF}) is mainly determined by the Jacobian $\textbf{J} =\partial \textbf{f}/\partial X_F + \textbf{f}_{\Omega} $. According to Eq.~(\ref{con:NonlinearXF}), 
\begin{equation}
\begin{aligned}\label{con:Jacobian}
\textbf{J} \left( X_F \right)
= \begin{bmatrix}
\kappa_n R_n - \gamma & -\kappa_n I_n +\left( \Omega - \tilde{\omega}_{n}\right) \\
\kappa_n I_n -\left( \Omega - \tilde{\omega}_{n}\right) & \kappa_n R_n - \gamma
\end{bmatrix},
\end{aligned}
\end{equation}
which reduces to $\partial \textbf{f}/\partial X_F$ when $\Omega = \tilde{\omega}_{n}$.

For the nonlinear system in Eq.~(\ref{eq:FRealI}),  we regard $u = \Omega - \tilde{\omega}_{n}$ as an unknown input, representing a constant perturbation or random unknown uncertainty due to the fact that the environment cannot be precisely modeled. Then
\begin{equation} \label{con:NonlinearXFControl}
\begin{aligned}
\frac{\mathrm{d}X_F(t)}{\mathrm{d} t} 
&=  \textbf{f}\left(X_F(t)\right) +\begin{bmatrix}
0 & u \\
-u & 0
\end{bmatrix}  X_F(t) \\
& = \left( \mathbf{A}\left(X_F(t)\right) + \mathbf{B} u\right)X_F(t)  + \begin{bmatrix}
\gamma\chi_n/2 \\
0
\end{bmatrix},
\end{aligned}
\end{equation}
where $\textbf{f}\left(X_F(t)\right) =\mathbf{A}\left(X_F(t)\right) X_F(t) + \begin{bmatrix}
\gamma\chi_n/2 \\
0
\end{bmatrix}$, $\mathbf{A}$ is the first matrix on the RHS of Eq.~(\ref{con:NonlinearXF}), $\mathbf{B} = \begin{bmatrix}
0 & 1 \\
-1 & 0
\end{bmatrix}$, and we define the output of the system as
\begin{equation} \label{con:OutputNonlinear}
\begin{aligned}
y_F(t) &= \begin{bmatrix}
1 & 1 \\
1 & -1
\end{bmatrix} \begin{bmatrix}
R_n(t) \\
I_n(t)
\end{bmatrix} \triangleq \mathcal{C} X_F(t).
\end{aligned}
\end{equation}

According to Refs.~\cite{Automatica1998contraction,Automatica2008contraction}, the Lyapunov function can be defined as $V\left(X_F\right) = \textbf{f}^{\rm T}\left(X_F\right) \mathbf{M}\left (X_F \right ) \textbf{f}\left(X_F\right)$, where $\mathbf{M}\left (X_F \right ) $ is a contraction matrix to be determined. When $u = 0$,
\begin{equation} \label{con:Vkdot}
\begin{aligned}
\dot{V}(X_F)=& \textbf{f}^{\rm T}\left(X_F(t)\right) \left [\frac{\partial \textbf{f}^{\rm T}}{\partial X_F} \mathbf{M}\left (X_F \right ) + \mathbf{M}\left (X_F \right ) \frac{\partial \textbf{f}}{\partial X_F} + \dot{\mathbf{M}}  \right] \textbf{f}\left(X_F(t)\right).
\end{aligned}
\end{equation}

For the simplest case with $\Omega= \tilde{\omega}_{n}$ and $I_n(t) \equiv 0$, Eq.~(\ref{con:NonlinearXF}) reduces to 
$\mathrm{d} R_n/\mathrm{d} t= \kappa_n  R_n^2 - \gamma R_n+ \gamma \chi_n/2$. More explanations are given combined with the following example and propositions.

\subsubsection{Example}
Take the parameter settings in Section~\ref{Sec:ParaExample}) as an example. When $\Omega = \tilde{\omega}_{n}$ and $ 4\kappa_nS_n= Q_n^2$, $I_n(t) \equiv 0$, $Q_n = -\gamma$, $S_n = \gamma^2/4\kappa_n$, $R_n(t) = \gamma/2\kappa_n - \gamma/\left(\gamma \kappa_nt + 2\kappa_n\right)$, then $\kappa_n R_n(t) - \gamma <0$ for arbitrary $t$ and $\textbf{J} \left( X_F(t) \right)$ is negative  definite. If we take $\mathbf{M} = \mathbf{I}$ as the identity matrix for simplification in Eq.~(\ref{con:Vkdot}), then $\dot{V}(X_F) < 0$ and Eq.~(\ref{eq:FRealI}) is globally stable in this parameter setting.

Generalized from the definition of invariant set as well as bounded-input bounded-output (BIBO) stability defined and introduced in Appendix~\ref{sec:NonlinearAppendix}, and the relative property in Ref.~\cite{lasalle1960some}, we can derive the following propositions.
\begin{mypro} \label{localinvariantset}
When $\Omega=\tilde{\omega}_n$ in Eq.~(\ref{con:NonlinearXF}) and $\mathbf{M} = \mathbf{I}$ in Eq.~(\ref{con:Vkdot}), there exists $\alpha_V > 0$ such that the parameters in Eq.~(\ref{con:NonlinearXF})  satisfy $2 \kappa_n \alpha_V + \gamma \chi_n - 2\gamma R_n<0$. In this case,  there exists an invariant set of $X_F$ represented as  $\Omega_{\alpha} = \left\{X_F \in \textbf{R}^2 : V \leq \alpha_V\right\}$  
 contained in the region satisfying $\dot{V} =0$, and  $\alpha_V \leq \gamma(\gamma-\chi_n\kappa_n)/\kappa_n^2$.
\end{mypro}

\begin{proof}
Assume initially $X_F \in \Omega_{\alpha}$ and $\Gamma_F$ represents an invariant set of $X_F$ satisfying that $\dot{V}(X_F) = 0$. When $2\kappa_n \alpha_V + \gamma \chi_n - 2\gamma R_n <0$, $\dot{V} < 0$ is satisfied according to Eqs.~(\ref{con:Vkdot}), (\ref{FRealILastV}) with $\Omega=\tilde{\omega}_n$, and $V(x) \geq 0$ by its definition. Then there exists a positive value $v_l$ that $\lim_{t \rightarrow \infty} V(X_F) = v_l \geq0$ and $v_l < \alpha_V$. Thus $V(X_F)$ is always within the invariant set. Denote the solutions of $R_n$ satisfying $\dot{V} = 0$ as $r_1$ and $r_2$, then $\alpha_V \leq r_1^2 + r_2^2 = \gamma(\gamma-\chi_n\kappa_n)/\kappa_n^2$.
\end{proof}
\begin{remark}  \label{Sec:remarkParaCom}
\textbf{Proposition~\ref{localinvariantset}} means that larger $\gamma$ or smaller $\kappa_n$ can induce a larger invariant set with stronger Markovian property. This is why the Markovian approximation can be applied in the circumstance that the coupling between the quantum system and the environment is weak.
\end{remark}

\begin{mypro} \label{BIBOApply}
When $X_F$ in Eq.~(\ref{eq:FRealI}) is initially bounded and satisfies that $\dot{V} < 0$ when $\Omega=\tilde{\omega}_n$,
then the nonlinear process in Eq.~(\ref{eq:FRealI}) is BIBO stable.
\end{mypro}
\begin{proof}
We take $\mathbf{M} = \mathbf{I}$, $\dot{V}$ is independent of $\Omega-\tilde{\omega}_n$. Thus if  $\dot{V} < 0$ when $\Omega=\tilde{\omega}_n$, $\dot{V}<0$ also holds when $\Omega\neq\tilde{\omega}_n$ according to Eq.~(\ref{con:Vkdot}), then the system in Eq.~(\ref{eq:FRealI}) is BIBO stable based on Proposition~\ref{BIBOProposition} in Appendix~\ref{sec:NonlinearAppendix}.
\end{proof}

\begin{remark}
When nonlinear process in Eq.~(\ref{eq:FRealI}) is BIBO stable and the Jacobian in Eq.~(\ref{con:Jacobian}) is bounded, the Lyapunov function can also be bounded according to Ref.~\cite{ruffer2013convergent}.
\end{remark}

\begin{lemma} \label{FRlim}
For the coupled nonlinear dynamics in Eq.~(\ref{eq:FRealI}), when $\dot{V} < 0$ in Eq.~(\ref{con:Vkdot}),   $\lim_{t\rightarrow \infty}R_n(t)$ and $\lim_{t\rightarrow \infty}I_n(t)$ both exist when $\left |\Omega - \tilde{\omega}_{n} \right | \leq \epsilon$.
\end{lemma}

\begin{proof}
When $\dot{V} < 0$ and $V(t) > 0$, $\lim_{t\rightarrow \infty}V(t)$ exists and $\lim_{t\rightarrow \infty}\dot{V}(t) = 0$. Then $\lim_{t\rightarrow \infty}R_n(t)$ and $\lim_{t\rightarrow \infty}I_n(t)$ can be solved by Eq.~(\ref{con:Vkdot}).
\end{proof}

\begin{mypro} \label{FpropertyDetun}
If the nonlinear dynamics in Eq.~(\ref{eq:FRealI}) is BIBO stable when $\Omega = \tilde{\omega}_{n}$, there exists $\epsilon > 0$ such that the trajectory of $X(t)$ in Eq.~(\ref{eq:NlevelOperatorEquation}) converges to be Markovian when $t\rightarrow \infty$ if $\left |\Omega - \tilde{\omega}_{n} \right | \leq \epsilon$.
\end{mypro}
\begin{proof}
Similar to Proposition~\ref{BIBOApply}, the Lyapunov function in Eq.~(\ref{con:Vkdot}) is independent of $\left ( \Omega - \tilde{\omega}_n\right)$. When the system is BIBO stable and $\left |\Omega - \tilde{\omega}_{n} \right | \leq \epsilon$, then $R_n(t)$ and $I_n(t)$ are solvable with steady values according to Lemma~\ref{FRlim} and Eq.~(\ref{con:Jacobian}). Then the trajectory of $X(t)$ in Eq.~(\ref{eq:NlevelOperatorEquation}) converges to be Markovian because $R_n$ and $I_n$ reach constants when $t$ is large enough.
\end{proof}

Taking the circumstance $N=2$ as an example, the real and imaginary parts of $F_1(t)$ are compared in Fig.~\ref{fig:MarkoviaParaCom}, where we take $\chi_1 = 1$, $\gamma = 2$GHz~\cite{NatureExperiment}, $\kappa_1=1$GHz, $\Omega= 50$GHz and $\mathrm{d} t = 0.01$ns. The comparisons in  Figs.~\ref{fig:MarkoviaParaCom}(a) and (b) show that, the diverse between $\Omega$ and $\tilde{\omega}_1$ can induce oscillations of $R_1(t)$ and $I_1(t)$, and this can influence the non-Markovian dynamics via the time-varying decay rates in Eq.~(\ref{eq:NlevelOperatorEquation}).

\begin{figure}[h]
\centerline{\includegraphics[width=1\columnwidth]{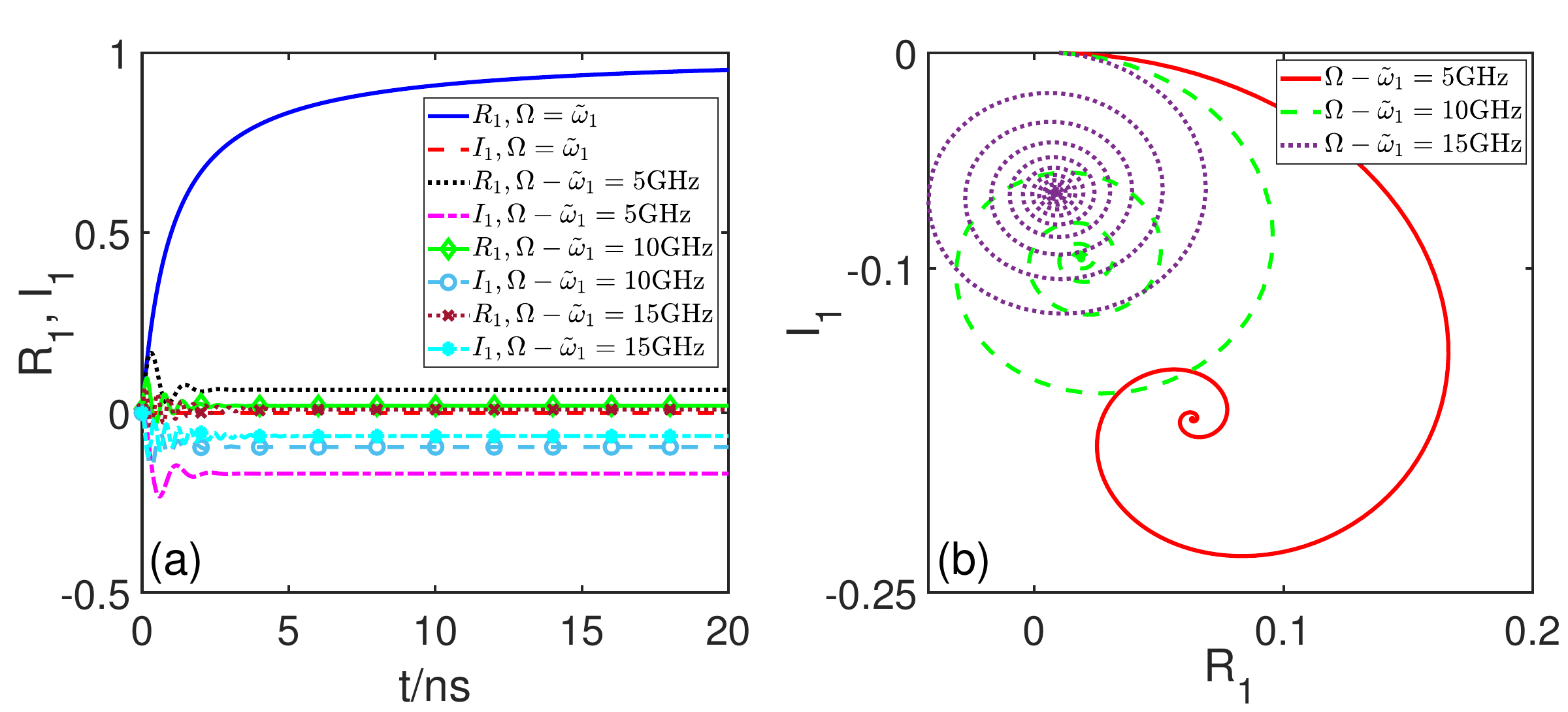}}
\caption{Compare different parameter settings for non-Markovian interactions.}
	\label{fig:MarkoviaParaCom}
\end{figure}

However, for the multi-level system in Fig.~\ref{fig:NonMarkovian}, the condition that $\tilde{\omega}_{n} = \Omega$ in Theorem~\ref{Fproperty} can never be simultaneously satisfied for all the energy levels labeled by $n$, thus we need to generalize Theorem~\ref{Fproperty} to the circumstance that $\tilde{\omega}_{n} \neq \Omega$ by regarding their difference as uncertain inputs as follows.

\subsection{Nonlinear dynamics with an uncertain input}
Practically, we usually cannot know the exact value of the environment parameter $\Omega$, thus the component $\Omega - \tilde{\omega}_n$ in Eq.~(\ref{eq:FRealI}) can be regarded as an uncertain input to the nonlinear system. Then Eq.~(\ref{con:NonlinearXF}) can be re-written as  
\begin{equation} \label{eq:unknownRandomInout}
\begin{aligned}
\frac{\mathrm{d}}{\mathrm{d} t} X_F(t)  &= \textbf{f}\left(X_F(t)\right) + \textbf{u},
\end{aligned}
\end{equation}
where we denote $\textbf{u} = \Delta A X_F(t) $ with $\Delta A  = \begin{bmatrix} 0 & \Omega- \tilde{\omega}_{n}\\ \tilde{\omega}_{n} - \Omega &0\end{bmatrix}$, and there always exists $\epsilon>0$ satisfying  $\| \Delta A \| \leq \epsilon$. Then Eq.~(\ref{eq:unknownRandomInout}) can be regarded as a perturbed nonlinear system. 
\begin{assumption} \label{Idealstable}
The nonlinear equation (\ref{eq:FRealI}) with $\Omega = \tilde{\omega}_n$ is stable according to the parameter settings in Theorem~\ref{Fproperty}.
\end{assumption}
According to Eqs.~(\ref{con:rewriteF})-(\ref{con:F3}), when $\Omega = \tilde{\omega}_{n}$, $I_n(t)\equiv 0$, and we denote $\lim_{t\rightarrow \infty} R_n(t) = \bar{R}_n$. Define a new state vector $\tilde{X}_F(t) = \left [R_n(t)-\bar{R}_n, I_n(t)\right]^{\rm T} \triangleq \left [\tilde{R}_n(t), I_n(t)\right]^{\rm T}$ and $\lim_{t\rightarrow \infty} \tilde{X}_F(t) = \left [0,0\right]^{\rm T}$ when $\Omega = \tilde{\omega}_n$ and Assumption~\ref{Idealstable} is satisfied. Then
\begin{subequations} \label{eq:FChange}
\begin{numcases}{}
\frac{\mathrm{d}}{\mathrm{d} t}\tilde{R}_n(t) =  \kappa_n \left[ \left(\tilde{R}_n(t) + I_n\right)^2  - I_n^2(t) \right]  - \gamma  \left(\tilde{R}_n(t) + \bar{R}_n\right) 
+ \left( \Omega -\tilde{\omega}_{n}\right)I_n(t) + \frac{\gamma \chi_n}{2}, \\
\frac{\mathrm{d}}{\mathrm{d} t} I_n(t) = 2 \kappa_n  \left(\tilde{R}_n(t) + \bar{R}_n\right) I_n(t) -  \left( \Omega - \tilde{\omega}_{n}\right) \left(\tilde{R}_n(t) + \bar{R}_n\right)  - \gamma I_n(t),
\end{numcases}
\end{subequations}
and can be simplified as
\begin{equation} \label{eq:XFtilde}
\begin{aligned}
\frac{\mathrm{d}}{\mathrm{d} t} \tilde{X}_F(t)  =& \tilde{\textbf{f}}\left(\tilde{X}_F(t)\right) + \begin{bmatrix} \frac{\gamma\chi_n}{2} +  \kappa_n \bar{R}_n^2 - \gamma \bar{R}_n\\ 
\left(\tilde{\omega}_{n} - \Omega \right)\bar{R}_n\end{bmatrix} + \begin{bmatrix} 0   & \Omega- \tilde{\omega}_{n}\\ 
\tilde{\omega}_{n} - \Omega & 0 
 \end{bmatrix}\tilde{X}_F(t),
\end{aligned}
\end{equation}
where $\tilde{\textbf{f}}\left(\tilde{X}_F(t)\right)$ is for the nonlinear component on the RHS of Eq.~(\ref{eq:FChange}). Obviously, $\tilde{\textbf{f}}\left(\tilde{X}_F(t)\right) = 0$ when $\tilde{X}_F = 0$. When $\Omega = \tilde{\omega}_{n}$ and $X_F(t)$ converges by Theorem~\ref{Fproperty}, the RHS of Eq.~(\ref{eq:XFtilde}) converges to zero.

\begin{theorem}
When $4\kappa_nS_n- Q_n^2  \leq 0$ in Eq.~(\ref{eq:FChange}), there exists $\epsilon > 0$ such that the trajectory of $X(t)$ in Eq.~(\ref{eq:NlevelOperatorEquation}) converges to a Markovian behavior when the uncertain $\Omega$ is bounded by $\left |  \tilde{\omega}_{n} - \Omega \right| \leq \epsilon$.
\end{theorem}

\begin{proof}
When $\tilde{\omega}_{n} = \Omega$, by Theorem~\ref{Fproperty}, the quantum dynamics approaches  Markovian behavior as $t\rightarrow \infty,$ and the nonlinear dynamics $\dot{\tilde{X}}_F(t)  = \tilde{\textbf{f}}\left(\tilde{X}_F(t)\right)$ is uniformly and asymptotically stable around its equilibrium $ \tilde{X}_F =0$. When $\left |  \tilde{\omega}_{n} - \Omega \right| \leq \epsilon$ and $\gamma\chi_n$ is finite,  Proposition~\ref{Chenstable} in Appendix~\ref{sec:NonlinearAppendix} ensures that $\tilde{R}_n(t)$ and $I_n(t)$ in Eq.~(\ref{eq:FChange}) remain stable. Then the trajectory of $X(t)$ in  Eq.~(\ref{eq:NlevelOperatorEquation}) converges to be Markovian as $t\rightarrow \infty$.
\end{proof}

The above theorem shows that, a trajectory of $X(t)$ in Eq.~(\ref{eq:NlevelOperatorEquation}) can converge to be Markovian even when $\Omega \neq \tilde{\omega}_n$ and $\Omega$ is not precisely known, if only the uncertain $\Omega$ is bounded around $\tilde{\omega}_{n}$.

Now we take $N=3$ as an example for numerical simulations of the above non-Markovian dynamics around the parameter settings in Refs.~\cite{fink2008climbing,fink2010quantum}. As shown in Fig.~\ref{fig:threelevel}, we take $\omega_1 = 37.7$GHz, $\omega_2 = 75.3$GHz, $\Omega = 50$GHz, $\Delta_1 = -20$MHz, $\Delta_2 = -70$MHz, $g_1 = g_2 = 20$MHz, $\chi_1= \chi_2 = 1$, $\gamma = 10$GHz, $\kappa_1 = \kappa_2 = 0.31$GHz, and the time step for simulation equals $0.01$ns. In the simulations for Markovian circumstance, we take $F_n$ with $n=1,2$ as constants that equal the steady values of those in the non-Markovian simulations.
\begin{figure}[h]
\centerline{\includegraphics[width=0.7\columnwidth]{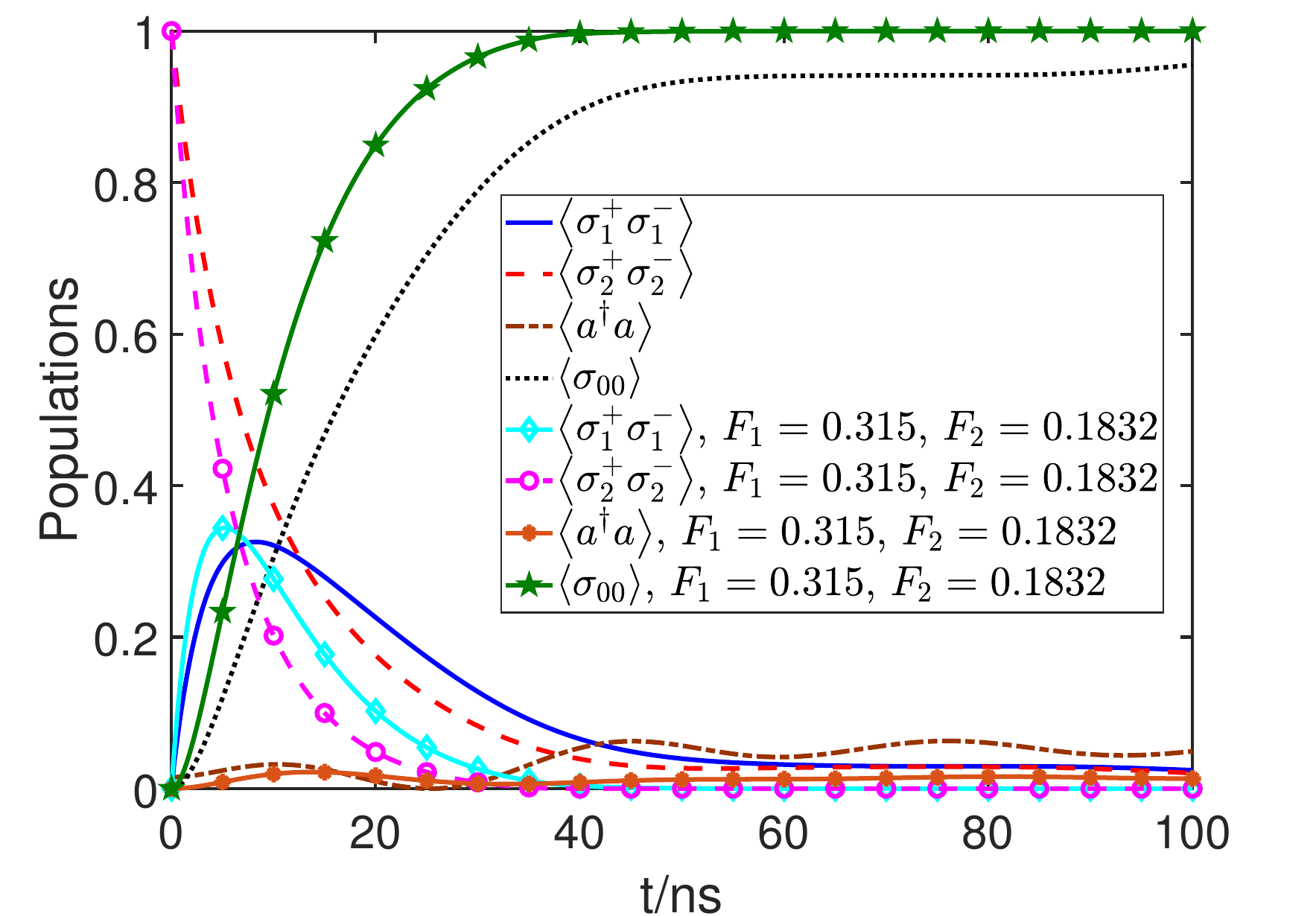}}
\caption{Non-Markovian and Markovian dynamics based on the quantum system that one three-level atom is coupled to a cavity.}
	\label{fig:threelevel}
\end{figure}

\section{Time-varying linear quantum control dynamics} \label{Sec:quantumcontrolOpen}
After clarifying the nonlinear parameter dynamics in the section above, in this section, we study the quantum open-loop control and measurement feedback control based on the control model in Eq.~(\ref{eq:NlevelOperatorEquation}), and the results are specially influenced by the nonlinear parameter dynamics in Eq.~(\ref{NMEqF}) due to the non-Markovian interactions between the quantum system and environment.

In the following, we only consider the circumstance that $4\kappa_nS_n- Q_n^2\leq  0$, thus the interaction between the atom and environment finally converges to be Markovian. 
Return to Eq.~(\ref{con:TimevariedEquation}), we analyze how the quantum dynamics can be influenced by detunings between the atom and cavity, and drive fields applied upon the quantum system.

For the simplified circumstance that $\Delta_n = 0$ as in Refs.~\cite{jing2010non,wang2018adiabatic}, then $\tilde{g}_n(t) \equiv g_n$, $C_n$ and $D_n$ in Eqs.~(\ref{con:TVECnt}) and (\ref{con:TVDCnt}) are all time-invariant for arbitrary $n$, while $A_n$ is time-varying with $F_n(t)$ converges to its steady values according to Eqs.~(\ref{con:F1}) and (\ref{con:F3}). $A(t)$ can be separated into the time-invariant component and time-varying components as
\begin{equation} \label{con:TVEAnt2}
\begin{aligned}
A_n(t)=&\bar{A}_n + \hat{A}_n(t)\\
=&\begin{bmatrix}
-L_n^{(1)} & -i \tilde{g}_n  & i \tilde{g}_n  \\
-i \tilde{g}_n & -L_n^{(2)} & 0\\
 i \tilde{g}_n  & 0 & -L_n^{(3)}
\end{bmatrix} + {\rm{diag}}\left(L_n^{(1)} -\left ( F_n + F_n^*\right), L_n^{(2)}-F_n^*,L_n^{(3)}-F_n  \right),
\end{aligned}
\end{equation}
where $L_n^{(1)} = \lim_{t\rightarrow \infty} \left ( F_n(t) + F_n^*(t)\right)$, $L_n^{(2)} = \lim_{t\rightarrow \infty}F_n^*(t)$ and $L_n^{(3)} =\lim_{t\rightarrow \infty} F_n(t)$, then $\lim_{t\rightarrow \infty} \hat{A}_n(t) = \textbf{0}_{3\times3}$. 

For the special case with $g_n = 0$, where the multi-level atom only interacts with the non-Markovian environment with $C_n = D_n^{\rm T} =\textbf{0}_{1\times3}$, we can construct the vector representing the populations of a multi-level atom as 
\[
\tilde{X}(t) = \left [\langle \sigma^+_{1} \sigma^-_{1}\rangle, \cdots,\langle \sigma^+_{n} \sigma^-_{n}\rangle,\cdots, \langle \sigma^+_{N-1} \sigma^-_{N-1}\rangle\right ]^{\rm T}.
\]
$\tilde{X}(t) $ is governed by the following real-valued equation if initially $\tilde{X}(0)$ is real,
\begin{equation} \label{con:populationdynamics}
\begin{aligned}
\frac{\mathrm{d}}{\mathrm{d} t} \tilde{X}(t) =  \left(\textbf{L}  + \textbf{P}(t)\right)\tilde{X}(t),
\end{aligned}
\end{equation}
where $\textbf{P}(n,n) = L_n^{(1)} -\left ( F_n + F_n^*\right)$, $\textbf{P}(n,n+1) = \left ( F_{n+1} + F_{n+1}^*\right) -L_{n+1}^{(1)} $, all the other elements in $\textbf{P}$ equal zero, the time-invariant matrix $\textbf{L}$ can be determined by Eq.~(\ref{con:TVEAnt2}) and $\textbf{P}(t)$, and more details are omitted. We denote $\tilde{X}(t) = \Phi\left(t,t_0\right)\tilde{X}(t_0)$ for $t\geq t_0$,  and~\cite{Hu2020necessaryTAC} 
\begin{equation} \label{con:transitionM}
\begin{aligned}
\Phi\left(t,t_0\right) = e^{\textbf{L}\left(t-t_0\right)} + \int_{t_0}^t e^{\textbf{L}\left(t-\tau\right)}  \textbf{P}(\tau) \Phi\left(\tau,t_0\right) \mathrm{d}\tau,
\end{aligned}
\end{equation}
which is similar to the quantum dynamics with a time-varying control input~\cite{wu2023bicon}, then we can derive the following results.
\begin{lemma}
When the trajectories of $\tilde{X}(t)$  in Eq.~(\ref{con:populationdynamics}) with $g_n = 0$ converge to be Markovian, the time-varying component $\textbf{P}$ converges to zero, and the transition matrix in Eq.~(\ref{con:transitionM}) is bounded.
\end{lemma}
\begin{proof}
The convergence of quantum dynamics to a Markovian regime means the convergence of $F_n(t)$ according to Remark~\ref{RemarkMarkov}. Consequently, the time-varying components in Eq.~(\ref{con:populationdynamics}) converge to zero. In this case, $\lim_{t \rightarrow \infty} \textbf{P}(t) = \textbf{0}_{(N-1)\times (N-1)} $. Thus $\Phi\left(t,t_0\right)$ is bounded because $\lim_{t \rightarrow \infty} \partial_t \Phi\left(t,t_0\right) = 0$.
\end{proof}

Specially, when $\textbf{P}(t) \equiv \textbf{0}_{(N-1)\times (N-1)}$, Eq.~(\ref{con:populationdynamics})  reduces to a linear time-invariant system and is exponentially stable because $\textbf{L}$ is a diagonal matrix whose eigenvalues are all negative. 

Generally, $\tilde{X}$ can be solved as
\begin{equation} \label{con:Xsolution}
\begin{aligned}
\tilde{X}(t) = e^{ \textbf{L} t} \tilde{X}(0) + \int_0^t e^{\textbf{L} (t-\tau)}  \textbf{P}(\tau)  \tilde{X}(\tau) \mathrm{d} \tau,
\end{aligned}
\end{equation}
where the non-Markovian interaction between the atom and the environment can influence the integral kernel in Eq.~(\ref{con:Xsolution}), and further influence the convergence rate and steady value of $\tilde{X}$. The dynamics of the equation with the format of Eq.~(\ref{con:Xsolution}) has been introduced in Ref.~\cite{bellman2008stability}. Generalized from Theorem 2 in  Ref.~\cite{bellman2008stability}~(Chapter 2, Page 36), we have the following corollary.
\begin{corollary} \label{Xconverge}
Provided that $\|  \textbf{P}(t) \| \leq c_1$ for $t\geq t_0$ in Eq.~(\ref{con:Xsolution}), where $c_1$ is determined by $\max(L_n^{(1)})$, it follows that $\lim_{t\rightarrow \infty } \tilde{X}(t) = 0$ in Eq.~(\ref{con:Xsolution}).
\end{corollary}
Corollary~\ref{Xconverge} means that when the interaction between the atom and environment converges to a Markovian regime, $\tilde{X}(t)$ will converge to zero.  Besides, Corollary~\ref{Xconverge} is further generalized to  linear time-varying system in Ref.~\cite{bellman2008stability}, and this generalization will be used in the following analysis.

\subsection{$\Delta_n \neq 0$, $ g_n \neq 0$} 
In the following, we consider the most general case by separating $X(t)$ into its real and imaginary parts. Considering that in Eq.~(\ref{eq:NlevelOperatorEquation}), $ \langle \sigma^+_{n} \sigma^-_{n}\rangle$ and $\langle a^{\dag} a\rangle$ are both real, Eq.~(\ref{SPa}) and Eq.~(\ref{SMad}) are conjugate complex values, then we rewrite the state vector in a simplified format as 
\[
\mathbb{X}(t) = \left [\mathbb{X}_1(t), \mathbb{X}_2(t),\cdots, \mathbb{X}_{N-1}(t) , a^R\right ]^{\rm T},
\]
with the dimension $3N-2$,
where $\mathbb{X}_n(t) = \left[ X_n^R[1],X_n^R[2],X_n^I[2]\right]$, $X_n^R$ and $X_n^I$, $a^R$ and $a^I$ represent the real and imaginary parts of $X_n(t)$ and $\langle a^{\dag} a\rangle$ in Eq.~(\ref{con:XntDef}) respectively, $X_n^R[1]$ represents the first element of the vector $X_n^R$, and similar for other elements in the vector $\mathbb{X}(t)$.  $ X_n^I[1] = a^I = 0$, $X_n^R[3] = X_n^R[2]$ and $X_n^I[3] = -X_n^I[2]$. Then the evolution of $\mathbb{X}(t)$ reads
\begin{equation} \label{con:TimevariedEquation2}
\begin{aligned}
\frac{\mathrm{d}}{\mathrm{d} t} \mathbb{X} &= \begin{bmatrix}
\mathbb{A}_1(t) & 0 & \cdots & 0 & \mathbb{D}_1(t)\\
0 & \mathbb{A}_2(t) & \cdots & 0 & \mathbb{D}_2(t)\\
 \vdots &\vdots &\ddots  &\vdots  &\vdots\\
 0 & 0 & \cdots & \mathbb{A}_{N-1}(t) & \mathbb{D}_{N-1}(t)\\
\mathbb{C}_1(t) &\mathbb{C}_2(t) & \cdots  & \mathbb{C}_{N-1}(t) &0
\end{bmatrix} \mathbb{X}\\
&\triangleq  \mathbb{A}(t)\mathbb{X} ,
\end{aligned}
\end{equation}
which is a real-value equation. We take the $n$-th energy level as an example, denote $\hat{g}_n =g_n \sin\left(\Delta_n t\right)$, $\breve{g}_n =g_n \cos\left(\Delta_n t\right)$,  and $\dot{\mathbb{X}}_n(t)  = \mathbb{A}_n(t)\mathbb{X}_n(t) + \mathbb{D}_n(t)a^R$, then
\begin{equation} \label{con:TVECntt}
\begin{aligned}
&\mathbb{A}_n(t)=
\begin{bmatrix}
- 2\mathcal{R}_n(t)  & 2\hat{g}_n  &2\breve{g}_n\\
-\hat{g}_n  & - \mathcal{R}_n(t)  &\mathcal{I}_n(t)-\Delta_n \\
-\breve{g}_n & \Delta_n-\mathcal{I}_n(t)  &- \mathcal{R}_n(t) \\
\end{bmatrix},
\end{aligned}
\end{equation}
\begin{equation} \label{con:TVECnt1}
\begin{aligned}
\mathbb{C}_n(t)&=\begin{bmatrix}
0 & -2\hat{g}_n & -2\breve{g}_n
\end{bmatrix},
\end{aligned}
\end{equation}
\begin{equation} \label{con:TVDCntComplex}
\begin{aligned}
\mathbb{D}_n(t)&=\begin{bmatrix}
0 & -\hat{g}_n &-\breve{g}_n
\end{bmatrix}^{\rm T},
\end{aligned}
\end{equation}
where we take $ F_n(t)\kappa_n  \triangleq \mathcal{R}_n(t) + i \mathcal{I}_n(t)$ for simplification.

Upon this, when the interaction between the atom and the environment becomes asymptotically Markovian, we denote $\lim_{t\rightarrow \infty} \mathcal{R}_n(t) = \bar{\mathcal{R}}_n$ and $\lim_{t\rightarrow \infty} \mathcal{I}_n(t) = \bar{\mathcal{I}}_n$. Then Eq.~(\ref{con:TVECntt}) becomes 
\begin{equation} \label{con:TVECntCombine}
\begin{aligned}
\mathbb{A}_n(t)&=
\begin{bmatrix}
- 2\bar{\mathcal{R}}_n  & 2\hat{g}_n  &2\breve{g}_n\\
-\hat{g}_n  & - \bar{\mathcal{R}}_n &\bar{\mathcal{I}}_n-\Delta_n \\
-\breve{g}_n  & \Delta_n-\bar{\mathcal{I}}_n  &- \bar{\mathcal{R}}_n \\
\end{bmatrix} +\begin{bmatrix}
2\bar{\mathcal{R}}_n  - 2\mathcal{R}_n(t)   & 0  &0\\
0  &  \bar{\mathcal{R}}_n  -\mathcal{R}_n(t)   &\mathcal{I}_n(t)-\bar{\mathcal{I}}_n \\
0  &\bar{\mathcal{I}}_n-\mathcal{I}_n(t)  &\bar{\mathcal{R}}_n -\mathcal{R}_n(t)  \\
\end{bmatrix}\\
&=\bar{\mathbb{A}}_n(t) + \tilde{\mathbb{A}}_n(t),
\end{aligned}
\end{equation}
where $\bar{\mathbb{A}}_n$ represents the first matrix in Eq.~(\ref{con:TVECntCombine}) and $\tilde{\mathbb{A}}_n$ represents the second matrix. Obviously, $\bar{\mathbb{A}}_n(t)$ is periodic with $\bar{\mathbb{A}}_n(t) = \bar{\mathbb{A}}_n\left(t+2\pi/\Delta_n\right)$, and the norm of $ \tilde{\mathbb{A}}_n(t)$ finally converges to zero.

As a combination of the conclusion in Ref.~\cite{vrabel2019note} on the relationship between $\Pi^+(t)$ (see definition in \cite{vrabel2019note} or Appendix~\ref{sec:NonlinearAppendix}), the stability of a periodic linear time-varying system, and the generalization of Corollary~\ref{Xconverge} to the linear time-varying system in Ref.~\cite{bellman2008stability}, we can derive the following corollary for the linear time-varying system in Eq.~(\ref{con:TimevariedEquation2}) with the time-varying matrix in Eq.~(\ref{con:TVECntCombine}).
\begin{corollary} \label{TimevaryPeriodCombine}
For the linear time-varying system (\ref{con:TimevariedEquation2}) with the matrix represented as a combination of time-varying periodic and time-varying converging matrices, as given in Eq.~(\ref{con:TVECntCombine}), the solutions approach zero when\\
(a) $\bar{\Pi}^+(t_0+T)  = \int_{t_0}^{t_0+T} \mu \left[\mathbb{A}(\tau)\right] \mathrm{d}\tau< 0 $,\\
(b) $\int_0^{\infty} \| \tilde{\mathbb{A}}_n(t) \| \mathrm{d} t < \infty $ for arbitrary $n$. 
\end{corollary}
\begin{proof}
See Appendix~\ref{sec:linearAppendix}.    
\end{proof}

\subsection{With external drives}
When there is a drive field with the amplitude $\mathcal{E}$ applied upon the cavity, the Hamiltonian $\bar{H}_I$ in Eq.~(\ref{con:TwoNlevelatomHam}) should be replaced by
\begin{equation} \label{con:TwoNlevelatomHamDrive}
\begin{aligned}
H' &=  \bar{H}_I  -i \mathcal{E} \left(a-a^{\dag}\right).
\end{aligned}
\end{equation}

We denote $R_a = \left \langle a+a^{\dag} \right \rangle$, $I_a = i \left\langle a-a^{\dag}\right\rangle$ with real values, $s_n^R$ and $s_n^I$ represent the real and imaginary parts of $\left\langle\sigma^-_n\right \rangle$, respectively. Then generalized from the state vector $\mathbb{X}(t)$ in Eq.~(\ref{con:TimevariedEquation2}), we define another real-valued vector $\tilde{\mathbb{X}}$ as
\[
\tilde{\mathbb{X}}(t) = \left [\tilde{\mathbb{X}}_1(t), \tilde{\mathbb{X}}_2(t),\cdots, \tilde{\mathbb{X}}_{N-1}(t) , a^R, R_a, I_a \right ]^{\rm T},
\]
where $\tilde{\mathbb{X}}_n(t)$ is generalized from $\mathbb{X}_n(t)$ in Eq.~(\ref{con:TimevariedEquation2}) as
\[
\tilde{\mathbb{X}}_n(t) = \left[ X_n^R[1],X_n^R[2],X_n^I[2],s_n^R,s_n^I\right],
\]
and $a^R$ has been defined in Eq.~(\ref{con:TimevariedEquation2}). Then the evolution of $\tilde{\mathbb{X}}(t)$ reads 
\begin{equation} \label{con:TimevariedEquationDrive}
\begin{aligned}
\frac{\mathrm{d}}{\mathrm{d} t} \tilde{\mathbb{X}} &=  \mathbb{B}(t)\tilde{\mathbb{X}} + 2 \mathcal{E} \begin{bmatrix}  0, 0 , \cdots,
1, 0\end{bmatrix}^{\rm T},
\end{aligned}
\end{equation}
where 
\begin{equation}
\begin{aligned}
\mathbb{B}(t) = \begin{bmatrix}
\mathbb{B}_1(t) & 0 & \cdots & 0 & \tilde{\mathbb{D}}_1(t)\\
0 & \mathbb{B}_2(t) & \cdots & 0 & \tilde{\mathbb{D}}_2(t)\\
 \vdots &\vdots &\ddots  &\vdots  &\vdots\\
 0 & 0 & \cdots & \mathbb{B}_{N-1}(t) & \tilde{\mathbb{D}}_{N-1}(t)\\
\tilde{\mathbb{C}}_1(t) &\tilde{\mathbb{C}}_2(t) & \cdots  & \tilde{\mathbb{C}}_{N-1}(t) &\tilde{\mathbb{P}}
\end{bmatrix},
\end{aligned}
\end{equation}

\begin{equation} \label{con:TVEBntdrive}
\begin{aligned}
&\mathbb{B}_n=\begin{bmatrix}
- 2\mathcal{R}_n   & 2\hat{g}_n  &2\breve{g}_n& 0 & 0\\
-\hat{g}_n  & - \mathcal{R}_n   &\mathcal{I}_n-\Delta_n &\mathcal{E} &0 \\
-\breve{g}_n  & \Delta_n-\mathcal{I}_n  &- \mathcal{R}_n &0 &-\mathcal{E}\\
0  & 0  &0 &- \mathcal{R}_n  &\mathcal{I}_n\\
0  & 0  &0 &- \mathcal{I}_n &- \mathcal{R}_n 
\end{bmatrix},
\end{aligned}
\end{equation}
\begin{equation} \label{con:TVECntdrive}
\begin{aligned}
\tilde{\mathbb{C}}_n&=\begin{bmatrix}
0 & -2\hat{g}_n & -2\breve{g}_n &0 &0\\
0 & 0 &0 &-2\hat{g}_n  &2\breve{g}_n \\
0 & 0 &0 &2\breve{g}_n  &2\hat{g}_n 
\end{bmatrix},
\end{aligned}
\end{equation}
\begin{equation} \label{con:TVEDntdrive}
\begin{aligned}
\tilde{\mathbb{D}}_n&=\begin{bmatrix}
0 &  0 &0\\ 
-\hat{g}_n & 0 & 0 \\
-\breve{g}_n &0 & 0 \\
0 & \frac{\hat{g}_n}{2} &- \frac{\breve{g}_n}{2} \\
0 & -\frac{\breve{g}_n}{2} & \frac{\hat{g}_n}{2} 
\end{bmatrix},
\end{aligned}
\end{equation}
and
\begin{equation} \label{con:Pt3M3}
\begin{aligned}
\tilde{\mathbb{P}}&=\begin{bmatrix}
0 &  \mathcal{E}  &0\\ 
0 & 0 & 0 \\
0 &0 & 0 
\end{bmatrix}.
\end{aligned}
\end{equation}

Above all, when there are no external drives, namely $\mathcal{E} = 0$, the initially excited atom can decay to the ground state and the cavity will finally be empty, which can also be interpreted with the quantum system's exponential stability similarly as in Ref.~\cite{ding2023quantumNlevel} by regarding the components outside of the cavity as an environment. However, when $\mathcal{E} > 0$, the exponential stability explaining the atom and cavity's decaying can be destructed, especially when $\mathcal{E}$ is large compared with the non-Markovian exchanging rate between the atom and the environment.

\begin{figure}[h]
\centerline{\includegraphics[width=1\columnwidth]{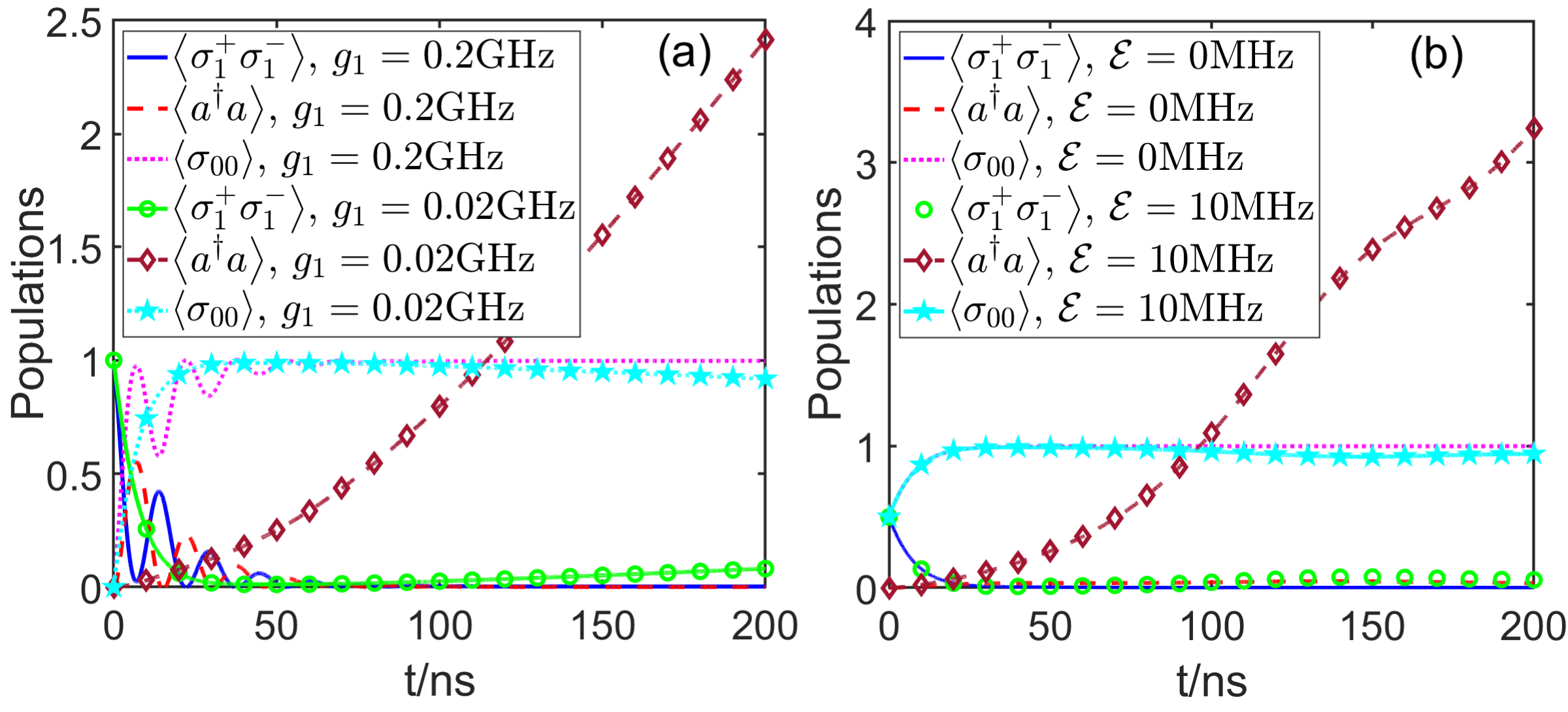}}
\caption{Non-Markovian dynamics of the quantum system with one two-level atom coupled to the cavity with applied drives.}
	\label{fig:TwolevelDrive}
\end{figure}
We take the two-level atom as an example to illustrate how the dynamics can be influenced by the non-Markovian environment, coupling strength between the atom and cavity, and the external drive applied on the cavity, which are compared in Fig.~\ref{fig:TwolevelDrive}. In all the simulations, $\omega_1 = 37.7$GHz, $\kappa_1 = 0.31$GHz, the time step for simulations and the non-Markovian environmental parameter settings are the same as in Fig.~\ref{fig:threelevel}. In Fig.~\ref{fig:TwolevelDrive}(a), the cavity is resonant with the atom, and we compared how the dynamics is influenced by the coupling strength $g_1$ when $\mathcal{E} = 10$MHz. It can be seen that when the cavity is driven, there can be multi-photon states in the cavity if the coupling strength between the atom and cavity is weak. In Fig.~\ref{fig:TwolevelDrive}(b), the atom is initially in a superposition state, $g_1 = 20$MHz, $\Delta_1 = -20$MHz, there can also be multi-photon states when the cavity is driven.

\section{Quantum measurement feedback control in the non-Markovian cavity-QED system}  \label{Sec:OneCavityMF}
For general cavity-QED systems, the decaying process of both the atom and cavity to the environment can be regarded as Markovian or non-Markovian. The non-Markovian interaction between the cavity and environment can be generalized from Appendixes~\ref{sec:NonMarkovParemeters} and \ref{sec:NonMarkovMaster} by replacing the atomic operator with the cavity operator. We further clarify the quantum measurement feedback control in this section based on the above generalization.

When the quantum system is measured via the cavity output as in Fig.~\ref{fig:NonMarkovian}, the feedback control can be designed according to the measurement result~\cite{zhang2017quantum,wiseman1994quantum}
\begin{equation} \label{eq:measureFeedback}
\begin{aligned}
I_c(t) & = \sqrt{2} \langle x \rangle + \frac{\xi(t)}{\sqrt{\eta}},
\end{aligned}
\end{equation}
where $x=\left( a+a^{\dag}\right)/\sqrt{2}$,  $\xi(t)$ is the measurement-induced random term satisfies that $\langle \xi(t) \rangle = 0$ and $\langle \xi(t) \xi(t') \rangle = \delta(t-t')$, and $\eta$ is the measurement detection efficiency. The feedback Hamiltonian reads
\begin{equation} \label{eq:FeedbackHam}
\begin{aligned}
H_{\rm fb}(t) &= I_c(t)G,
\end{aligned}
\end{equation}
where $G = G^{\dag}$ is the feedback operator. We assume that the detection efficiency of the measurement apparatus is ideal as $\eta = 1$, and see Ref.~\cite{van2005modelling} for the case $\eta <1$ with a similar format.

Generalized from Eq.~(\ref{con:SSENonMarUpdate0}) and the Markovian stochastic master equation (SME) with measurement feedback control in Refs.~\cite{ahn2002continuous,ticozzi2008TAC}, 
the measurement feedback based SME when there are non-Markovian interactions between the quantum system and environment reads~\cite{wiseman1994quantum}
\begin{equation} \label{con:NMFeedback}
\begin{aligned}
\mathrm{d}\rho =&-i\left[\bar{H}+H_D,\rho\right]\mathrm{d} t +\sum_{n=1}^{N-1}\left[\left( F_n + F_n^*\right)   L_n  \rho L_n^{\dag} -   \left( F_n  L_n^{\dag}L_n\rho + F_n^* \rho L_n^{\dag}L_n\right)\right]\mathrm{d} t \\
&+\kappa\left( F_a + F_a^*\right)   a  \rho a^{\dag}\mathrm{d} t  - \kappa\left( F_a  a^{\dag}a\rho + F_a^* \rho a^{\dag}a\right)  \mathrm{d} t \\
& - i g_f  \left [G, a\rho +  \rho a^{\dag}\right ] \mathrm{d}t -\frac{g_f^2}{2} \left [ G , \left [G, \rho\right ] \right ]\mathrm{d} t +\mathcal{H}[a]\rho \mathrm{d}W  - ig_f   \left [G, \rho\right ] \mathrm{d} W,
\end{aligned}
\end{equation}
where $H_D$ represents external drives applied upon the cavity, $F_a$  is for the non-Markovian interaction between the cavity and the environment, and is governed by the similar dynamics as $F_n$ in Eq.~(\ref{con:FtEquation}) in Appendix~\ref{sec:NonMarkovParemeters} by respectively replacing the parameters $\omega_n$, $\kappa_n$ and $\chi_n$ with $\omega_c$, $\kappa_c$ for the coupling strength between the cavity and environment, and $\chi_a$ in analog to $\chi_n$ for the environment parameter setting, $\kappa$ represents a constant decay rate of the cavity to the environment,  $\mathcal{H}[O]\rho = O \rho + \rho O^{\dag} -  \rm{Tr}\left [  \right ( O + O^{\dag}\left )\rho \right ]\rho$ for an arbitrary operator $O$, $g_f$ is the feedback strength, $\mathrm{d}W(t) = \xi(t)\mathrm{d}t$ is for a Wiener process representing the Homodyne detection noise in the quantum measurement, $E\left[ \mathrm{d}W(t)^2 \right] = \mathrm{d}t$, $\xi(t)$ can be approximated with the white noise~\cite{gardiner1985handbook}, and the meaning of the other components is the same as Eq.~(\ref{con:SSENonMarUpdate0}).

In the following, we firstly clarify the circumstance that the feedback operator $G$ is a cavity operator, then illustrate the circumstance that $G$ is an atomic operator with numerical simulations.

\subsection{Feedback control with cavity operators} \label{Sec:cavityfeedback}
In this subsection, we assume that the interaction between the cavity and environment is Markovian with the decay rate $\kappa$. The feedback operator can be designed as~\cite{Kurt1999feedback}
\begin{equation} \label{eq:directFeedback}
\begin{aligned}
G &= \beta_x x + \beta_p p  =\frac{\beta_x - i \beta_p}{\sqrt{2}} a + \frac{\beta_x + i\beta_p}{\sqrt{2}} a^{\dag},
\end{aligned}
\end{equation}
where the position operator $x$ is given by Eq.~(\ref{eq:measureFeedback}) and the momentum operator $ p = i\left( a^{\dag} - a\right)/\sqrt{2}$ with $xp - px = i$. Besides, Eq.~(\ref{eq:directFeedback}) is in the rotating frame with respect to the feedback driving frequency $\omega_d$ from the original format $\tilde{G} =\left[ \left(\beta_x - i \beta_p\right) ae^{i\omega_d t} + \left(\beta_x + i\beta_p\right) a^{\dag} e^{-i\omega_d t}\right]/\sqrt{2}$~\cite{wang2017PRA}. Denote $\Delta = \omega_c - \omega_d$, then the dynamics of the mean-value of the operators $x$ and $p$ can be derived based on Eq.~(\ref{con:NMFeedback}) without $H_D$ as
\begin{subequations}  \label{eq:xpdynamics}
\begin{align}
&\langle \dot{x} \rangle = \Delta \langle p \rangle -\frac{\kappa}{2} \langle x \rangle  - i\sum_n \frac{\tilde{g}_n^* }{\sqrt{2}} \left \langle \sigma^-_n\right\rangle + i\sum_n \frac{\tilde{g}_n}{\sqrt{2}} \left \langle \sigma^+_n\right\rangle  + \sqrt{2}g_f \beta_p \langle x \rangle +\sqrt{2} \left(V_x -\frac{1}{2}\right)  \xi(t)  + g_f \beta_p  \xi(t), \label{xdot}\\
&\langle \dot{p} \rangle= -\Delta \langle x \rangle -\frac{\kappa}{2} \langle p \rangle  - \sum_n \frac{\tilde{g}_n^* }{\sqrt{2}} \left \langle \sigma^-_n\right\rangle - \sum_n \frac{\tilde{g}_n}{\sqrt{2}} \left \langle \sigma^+_n\right\rangle  - \sqrt{2}g_f  \beta_x\langle x\rangle  +\sqrt{2} V_{xp} \xi(t) -g_f \beta_x \xi(t), \label{pdot} \\
&\langle \dot{\sigma}^-_n \rangle = -i\tilde{g}_n   \langle a \rangle - F_n(t)\kappa_n\langle \sigma^-_n \rangle,\\
&\langle \dot{a} \rangle =-i\Delta\langle a \rangle -i\sum_n \tilde{g}_n^*  \langle \sigma^-_n \rangle - \frac{\kappa}{2}  \langle a \rangle  - ig_f \left ( \beta_x + i\beta_p\right) \langle x \rangle  \notag\\
&~~~~~~~-ig_f \left(\beta_x + i\beta_p \right)\xi(t)/\sqrt{2} +\left[ \left \langle \left(a+a^{\dag}\right)a \right \rangle - \left\langle a + a^{\dag}\right\rangle \left\langle a \right \rangle\right]\xi(t),
\end{align}
\end{subequations}%
where $V_x = \left \langle x^2\right \rangle  - \langle x \rangle^2 $ is the variance of the  position operator, $V_p$ is that for the momentum operator, and $V_{xp} = \left ( \langle xp \rangle +  \langle px \rangle\right)/2 - \langle x \rangle\langle p \rangle $ is the covariance between position and momentum~\cite{Kurt1999feedback}. Then based on the stochastic components in Eqs.~(\ref{xdot}) and (\ref{pdot}), we have the following proposition. 
\begin{mypro} \label{NoiseCancel}
The feedback with the cavity operator can cancel the noise when $g_f \beta_p = -\sqrt{2}\left( V_x - 1/2 \right)$ and $g_f\beta_x =  \sqrt{2} V_{xp}$.
\end{mypro}
\begin{proof}
The proposition can be easily derived by the noise components in Eqs.~(\ref{xdot}) and (\ref{pdot}), thus more details are omitted.
\end{proof}

Besides, $V_x$ and $V_{xp}$ are governed by the following nonlinear equation~\cite{jacobs2006straightforward,Kurt1999feedback,genoni2016conditional,zhang2010transition}
\begin{subequations} \label{eq:VarianceEp}
\begin{numcases}{}
\dot{V}_x  =2 \Delta V_{xp} - \kappa V_x' + 2\sqrt{2}g_f\beta_p V_x' +g_f^2 \beta_p^2  - \left ( \sqrt{2}V_x'  + g_f \beta_p \right)^2,  \label{Vxdot}\\
\dot{V}_{xp} = \Delta V_p - \Delta V_x - \kappa V_{xp} + \sqrt{2}g_f \beta_p V_{xp} - \sqrt{2} g_f \beta_x V_x  + \frac{\sqrt{2}}{2} g_f\beta_x -g_f^2 \beta_x\beta_p  -\left ( \sqrt{2}V_{xp}  - g_f \beta_x \right)\left ( \sqrt{2}V_x'  + g_f \beta_p \right),\label{Vxpdot}\\
 \dot{V}_p = -2\Delta V_{xp} - \kappa V_P + \frac{\kappa}{2} - 2\sqrt{2}g_f \beta_x V_{xp} + g_f^2 \beta_x^2    - \left ( \sqrt{2}V_{xp}  - g_f \beta_x \right)^2,\label{Vpdot}\
\end{numcases}
\end{subequations}
where we denote $V_x' = V_x - 1/2$, the last component of Eqs.~(\ref{Vxdot}), (\ref{Vxpdot}), and (\ref{Vpdot}) comes from the stochastic components in $\left ( \mathrm{d} \langle x \rangle  \right)^2$, $ \mathrm{d} \langle x \rangle \mathrm{d} \langle p \rangle,$ and  $\left ( \mathrm{d} \langle p \rangle  \right)^2$, respectively.

\begin{remark}
Proposition~\ref{NoiseCancel} as well as Eq.~(\ref{eq:VarianceEp}) reveals that the final steady states and the covariance can be modulated by choosing the measurement operator in Eq.~(\ref{eq:directFeedback}) with tunable $\beta_x$, $\beta_p$, and the feedback control field.
\end{remark}

For the open loop control without feedback, the variances in Eq.~(\ref{eq:VarianceEp}) converge to the steady values as $V_{x}(\infty)=V_{p}(\infty) = 0.5$ and $V_{xp}(\infty) = 0$, similar as in Ref.~\cite{zhang2010transition}.

\subsection{Feedback control with atomic operators}

The feedback operator for the $N$-level system can be represented as~\cite{wiseman2009quantum}
\begin{equation} \label{con:OneExciationVector}
\begin{aligned}
G = \sum_{i\neq j } G_{ij} = \sum_{i\neq j } \left (|i\rangle \langle j| + |j\rangle \langle i| \right),
\end{aligned}
\end{equation}
where $1 \leq i,j \leq N-1$.  In the following, we take the two-level atom case as an example.

When $N=2$, we choose $G = |0\rangle \langle 1| + |1\rangle \langle 0|$, and assume initially the atom is excited. The system's states can be evaluated with the vector of the operators $S = \left [ \langle a \rangle, \langle \sigma_- \rangle , \langle \sigma_z \rangle \right]^{\rm T} \triangleq \left [ \bar{\alpha}, \bar{s},  \bar{w} \right]^{\rm T}$, where $\sigma_- = |0\rangle \langle 1|$ and $\sigma_z = |1\rangle \langle 1| - |0\rangle \langle 0|$. When the external drive is $H_D = -i \mathcal{E} \left(a-a^{\dag}\right)$, the dynamics of the cavity or atomic operator reads
\begin{subequations} \label{eq:twolevelNumEquation}
\begin{numcases}{}
 \dot{\bar{\alpha}}=  -i \tilde{g}_1^*\bar{s} - F_a(t)\kappa \bar{\alpha}/2+\mathcal{E},\label{higherorderA}\\
 \dot{\bar{s} } =  i\tilde{g}_1  \bar{\alpha}\bar{w} - F_1(t) \kappa_1\bar{s} + i g_f \left (\bar{\alpha} + \bar{\alpha}^*  \right)  \bar{w} -g_f^2  \left(\bar{s} - \bar{s}^* \right) +ig_f  \bar{w} \xi(t),\\
 \dot{\bar{w}} = -2i\tilde{g}_1 \bar{s}^*\bar{\alpha} +2i \tilde{g}_1^*  \bar{s}\bar{\alpha} ^*  - \left ( F_1(t) + F_1^*(t)\right) \kappa_1 \left( \bar{w} +1\right)  - 2i g_f \left (\bar{s}^*  - \bar{s} \right)\left(\bar{\alpha} + \bar{\alpha}^* \right) -2g_f^2 \bar{w}   -2ig_f \left ( \bar{s}^*- \bar{s} \right ) \xi(t),
\end{numcases}
\end{subequations}
with the initially normalization condition $\bar{w}^2 + 4 \left | \bar{s} \right |^2 =1$~\cite{JCPRA}, and the stochastic component with higher order amplitude is omitted in Eq.~(\ref{higherorderA}). We consider the simplified steady states satisfying that $\dot{\bar{\alpha}} = \dot{\bar{s}} =\dot{\bar{w}} =0$ after averaging the measurement noise $\xi(t) = 0$, and take $\Delta_1 = 0$ in $\tilde{g}_1$ for simplification~\cite{zhang2010protecting,ticozzi2012stabilization}.

\begin{figure}[h]
\centerline{\includegraphics[width=1\columnwidth]{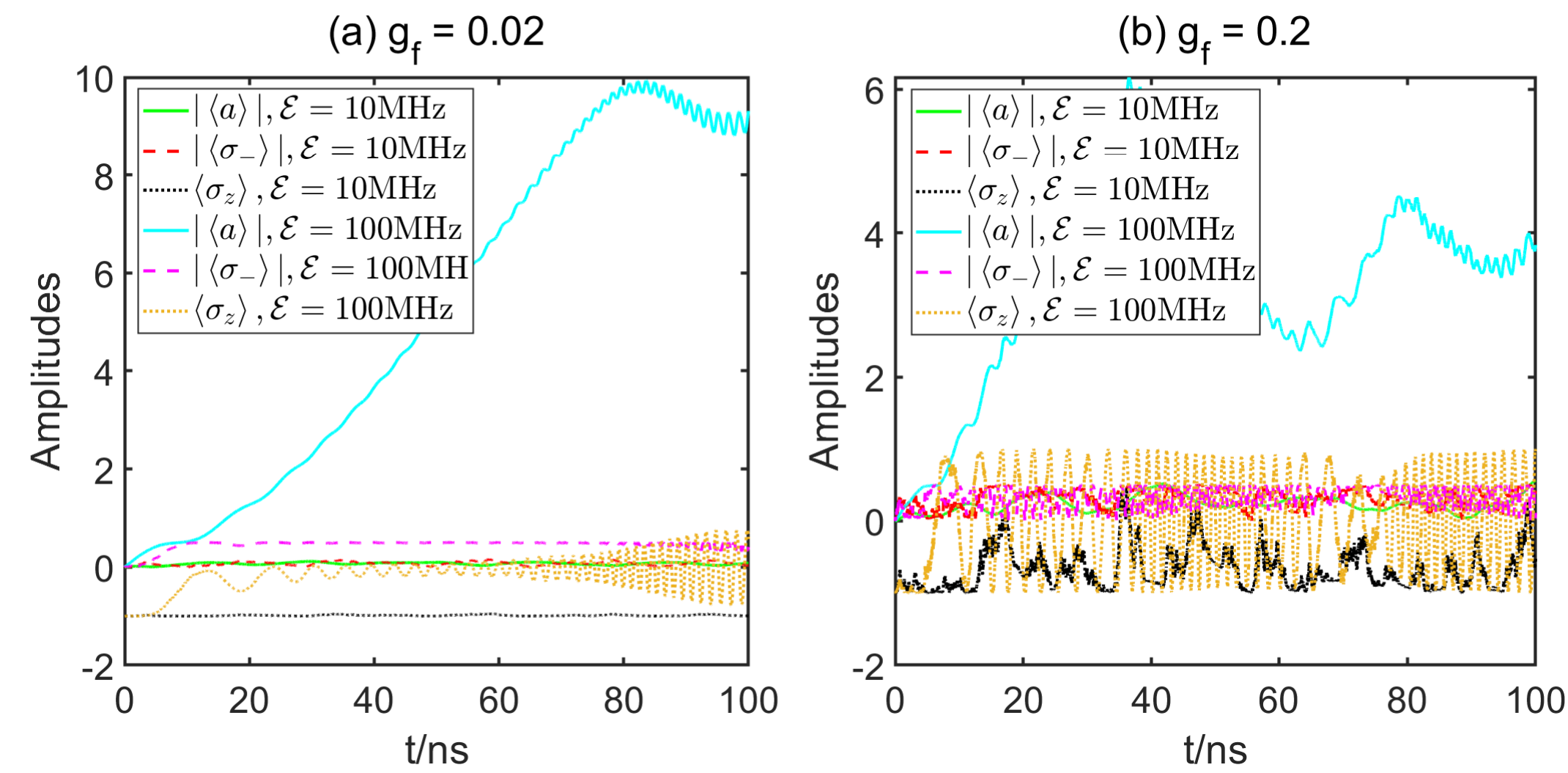}}
\caption{Measurement feedback control for one two-level atom in the cavity when both the atom and cavity interact with a non-Markovian environment.}
	\label{fig:OneatomFB}
\end{figure}

The performance for different non-Markovian parameter settings is compared in Fig.~\ref{fig:OneatomFB}, where we take $\kappa_1 = 1$GHz, $\kappa = 0.2$GHz, $\gamma = 2$GHz, $\chi = 1$, $\omega_a = 37.7$GHz, $\Omega = 32.7$GHz, $\tilde{g}_1 = 200$MHz and $\mathrm{d} t = 0.01$ns. It can be seen that the measurement feedback control can enhance the probability that the atom is excited, especially when the feedback strength and drive field are strong.

\section{Quantum control for the multiple Jaynes-Cummings model} \label{Sec:MultiCavity}
Based on the above analysis on the non-Markovian quantum dynamics with one cavity, in this section, we further study the non-Markovian dynamics that atoms are in an array of coupled cavities. When there are overall $M$ nearest neighbor coupled cavities with one two-level atom in each cavity, the free Hamiltonian for the atoms and cavities reads
\begin{equation} \label{con:McavityFree}
\begin{aligned}
H_M& =\sum_{m=1}^{M} \omega_c^{(m)} a^{\dag}_ma_m + \sum_{m=1}^{M} \omega_a^{(m)}\sigma_+^{(m)} \sigma_-^{(m)}  ,
\end{aligned}
\end{equation}
where $a_m$ ($a^{\dag}_m$) is the annihilation (creation) operator for the $m$-th cavity with the resonant frequency $\omega_c^{(m)}$, $\sigma_-^{(m)}$ ($\sigma_+^{(m)}$) represents the lowering (raising) operator for the two-level atom with frequency $\omega_a^{(m)}$ in the $m$-th cavity, and the detunings $\delta_m =  \omega_a^{(m)}- \omega_c^{(m)}$.

\begin{assumption} \label{AssumeMultiplyJC}
Assume the resonant frequencies and decay rates of the cavities are identical as $\omega_c^{(1)} = \cdots=\omega_c^{(M)}= \omega_c$, $\kappa_1 = \cdots = \kappa_M = \kappa$, while the resonant frequency of atoms $\omega_a^{(m)}$ in different cavities can be different. Initially the cavity is not empty such that $\left \langle a_m \right\rangle \neq 0$ and $\left \langle \sigma_-^{(m)}\right\rangle \neq  0$.   
\end{assumption}

The interaction Hamiltonian among the atoms and cavities reads
\begin{equation} \label{con:Mcavity}
\begin{aligned}
H_M& = \sum_{m=1}^{M} g_m\left (e^{-i\delta_m t}\sigma_-^{(m)} a^{\dag}_m + e^{i\delta_m t}\sigma_+^{(m)} a_m \right ) + \sum_m J_c \left (a^{\dag}_ma_{m+1} + a_{m}a^{\dag}_{m+1} \right),
\end{aligned}
\end{equation}
where $g_m$ is the coupling strength between the atom and cavity in the $m$-th cavity, $J_c$ is the coupling strength between two neighborhood cavities, the atom in each cavity is coupled to the non-Markovian environment similar as in Appendix~\ref{sec:NonMarkovParemeters}, 
and the atomic system is driven by the non-Markovian master equation as
\begin{equation} \label{con:MCavityMaster}
\begin{aligned}
&\dot{\rho} =-i\left[H_M,\rho\right ] + \sum_m \left [\kappa^{(m)}\left (\mathcal{F}_m^* + \mathcal{F}_m \sigma_-^{(m)}   \rho\sigma_+^{(m)} -  \mathcal{F}_m \sigma_+^{(m)}\sigma_-^{(m)}\rho
 - \mathcal{F}_m^*\rho \sigma_+^{(m)}\sigma_-^{(m)}\right) +  \kappa_{m}\mathcal{L}_{a_m} \left [ \rho \right ]  \right],%
\end{aligned}
\end{equation}
where $\sigma_-^{(m)} = I_1 \otimes \cdots \otimes I_{m-1} \otimes \sigma_- \otimes I_{m+1} \cdots \otimes I_{M}$, $\mathcal{L}_{O} = 2O\rho O^{\dag} -O^{\dag}O\rho -\rho O^{\dag}O$ for an arbitrary operator $O$,
\[
\frac{\mathrm{d} \mathcal{F}_m}{\mathrm{d} t} =\kappa^{(m)} \mathcal{F}_m^2  - \left(\gamma + i\Omega - i\omega_a^{(m)}\right)\mathcal{F}_m + \frac{\gamma \chi}{2},
\]
$\kappa^{(m)}$ represents the coupling between the atom in the $m$-th cavity and environment, $\gamma$ and $\chi$ are constants.

\subsection{Dynamics without drives}
Denote $\mathbf{x}_m =  \left [ \langle \sigma_-^{(m)} \rangle, \langle a_m \rangle \right]^{\rm T}$, and $\tilde{\mathbf{X}} = \left [ \mathbf{x}_1^{\rm T}, \mathbf{x}_2^{\rm T},\cdots,\mathbf{x}_M^{\rm T}\right]^{\rm T}$, then 
\begin{equation} \label{con:mthcavitydynamcis}
\begin{aligned}
\dot{\mathbf{x}}_m = &\begin{bmatrix}
- \mathcal{F}_m(t)\kappa^{(m)} &  -ig_m e^{i\delta_m t} \\
-ig_m e^{-i\delta_m t} & -\kappa_m
\end{bmatrix}\mathbf{x}_m -   \begin{bmatrix}
0\\
iJ_c
\end{bmatrix} \langle a_{m+1} \rangle -   \begin{bmatrix}
0\\
iJ_c
\end{bmatrix} \langle a_{m-1} \rangle.
\end{aligned}
\end{equation}
In the following, we denote and $\mathcal{F}_m(t) = \mathcal{F}^R_m(t) + i\mathcal{F}^I_m(t)$.

Based on Eq.~(\ref{con:mthcavitydynamcis}), we define the real-valued vector $\mathbf{\mathcal{X}}_M$ representing the real and imaginary components of $\tilde{\mathbf{X}}$ as
\[
\mathbf{\mathcal{X}}_M =  \left [ \mathbf{s}^R_1, \mathbf{s}^I_1, \cdots, \mathbf{s}^R_M, \mathbf{s}^I_M, \mathbf{a}^R_1, \mathbf{a}^I_1,  \cdots, \mathbf{a}^R_M, \mathbf{a}^I_M\right]^{\rm T},
\]
where $\mathbf{s}^R_m$ and $\mathbf{s}^I_m$ are the real and imaginary parts of $\langle \sigma_-^{(m)} \rangle$, $\mathbf{a}^R_m$ and $\mathbf{a}^I_m$ are the real and imaginary parts of $\langle a_m \rangle$, respectively, and $M \geq 2$. Then 
\begin{equation} \label{con:McavityRealImage}
\begin{aligned}
\dot{\mathbf{\mathcal{X}}}_M 
&= \begin{bmatrix}
\mathbf{F}(t) &\mathbf{G}(t) \\
\mathbf{R}(t)  &\mathbf{S}
\end{bmatrix}\mathbf{\mathcal{X}}_M \triangleq \mathbf{\mathcal{A}}_M  \mathbf{\mathcal{X}}_M,
\end{aligned}
\end{equation}
where
\begin{equation} \label{con:F}
\begin{aligned}
\mathbf{F}(t) = -{\rm diag}\left[ \mathbf{\mathcal{A}}^F_1(t)\kappa^{(1)}, \cdots, \mathbf{\mathcal{A}}^F_M(t)\kappa^{(M)}\right],
\end{aligned}
\end{equation}
\begin{equation} \label{con:G}
\begin{aligned}
\mathbf{G}(t) ={\rm diag} \left[\mathbf{\mathcal{A}}^G_1(t),\cdots,  \mathbf{\mathcal{A}}^G_M(t)\right],
\end{aligned}
\end{equation}
\begin{equation} \label{con:R}
\begin{aligned}
\mathbf{R}(t) = {\rm diag}\left[\mathbf{\mathcal{A}}^g_1(t),\cdots,  \mathbf{\mathcal{A}}^g_M(t)\right],
\end{aligned}
\end{equation}
\begin{equation} \label{con:S}
\begin{aligned}
\mathbf{S} 
& = I_M \otimes  \mathbf{\mathcal{A}}_{\kappa} + \begin{bmatrix}
0  & 1 & 0 &\cdots  & 0\\
1 &0 & 1&\cdots  & 0 \\
0 &1 &0 &\cdots   & 0 \\
\vdots &\vdots & \vdots  &\ddots  &\vdots\\
0 & 0 & 0 &\cdots  &0
\end{bmatrix}  \otimes  \mathbf{\mathcal{A}}_{J},
\end{aligned}
\end{equation}
with $\mathbf{\mathcal{A}}^F_m(t) = \begin{bmatrix} \mathcal{F}^R_m(t)  & - \mathcal{F}^I_m(t) \\ \mathcal{F}^I_m(t) & \mathcal{F}^R_m(t)\end{bmatrix}$, $\mathbf{\mathcal{A}}^G_m(t) = \begin{bmatrix} g_m \sin\left(\delta_m t\right)  & g_m \cos\left(\delta_m t\right) \\ -g_m \cos\left(\delta_m t\right) & g_m \sin\left(\delta_m t\right)\end{bmatrix}$, $\mathbf{\mathcal{A}}^g_m(t) = \begin{bmatrix} -g_m \sin\left(\delta_m t\right)  & g_m \cos\left(\delta_m t\right) \\ -g_m \cos\left(\delta_m t\right) & -g_m \sin\left(\delta_m t\right)\end{bmatrix}$, $\mathbf{\mathcal{A}}_{\kappa} = \begin{bmatrix}  -\kappa  & 0 \\ 0 &  -\kappa\end{bmatrix}$ and $\mathbf{\mathcal{A}}_{J} = \begin{bmatrix}  0  & J_c \\ -J_c &  0\end{bmatrix}$.

\begin{remark}
When the interaction between the atom and the environment is Markovian or converges from non-Markovian to Markovian interactions, $\mathcal{F}_m$ equals a constant and $\mathbf{\mathcal{A}}_M$ in Eq.~(\ref{con:McavityRealImage}) is periodic as $\mathbf{\mathcal{A}}_M\left(t+2\pi/\delta_m\right) = \mathbf{\mathcal{A}}_M\left(t\right)$ when $\delta_1 = \cdots = \delta_M\neq 0$. Besides, when $\delta_1 = \cdots = \delta_M = 0$, $\mathbf{\mathcal{A}}_M$ reduces to be time invariant.
\end{remark}

Obviously, the stability of the high-dimensional equation (\ref{con:McavityRealImage}), which is linear time-varying, can be similarly clarified by Corollary~\ref{TimevaryPeriodCombine}. For example, when there are two coupled cavities, the dynamics based on Markovian interactions between the quantum system and the environment has been analyzed in Refs.~\cite{wu2023bicon,zhang2019emission}. In this case, we denote $\mathbf{\mathcal{X}}_{\mathbf{a}} = \left[ \mathbf{a}^R_1, \mathbf{a}^I_1,  \mathbf{a}^R_2, \mathbf{a}^I_2 \right]^{\rm T}$, as a subspace of $\mathbf{\mathcal{X}}$. Obviously, when $g_m = 0$ for $m=1,2$,  
\begin{equation} \label{con:SubspaeX}
\begin{aligned}
\dot{\mathbf{\mathcal{X}}_{\mathbf{a}}} &= \begin{bmatrix}
\mathbf{\mathcal{A}}_{\kappa}  &\mathbf{\mathcal{A}}_{J} \\
\mathbf{\mathcal{A}}_{J}  & \mathbf{\mathcal{A}}_{\kappa}
\end{bmatrix}\mathbf{\mathcal{X}}_{\mathbf{a}} \triangleq \mathbf{\mathcal{A}}_{\mathbf{a}}\mathbf{\mathcal{X}}_{\mathbf{a}},
\end{aligned}
\end{equation}
which is linear and time-invariant.  $\mathbf{\mathcal{X}}_{\mathbf{s}}$ is exponentially stable because in the Laplace transformation based on Eq.~(\ref{con:SubspaeX}), the real part of roots determined by $(\textit{s}+\kappa)^2 + J_c^2 = 0$ are negative. Further by Corollary~\ref{Xconverge}, $\lim_{t\rightarrow \infty } \mathbf{\mathcal{X}}_{\mathbf{a}}(t) = 0$. This can be further generalized to the circumstance that $M>2$.

\subsection{Dynamics with drives}
The drive field can be applied via the feedforward or feedback channel. Here we take the measurement feedback as an example. 
Generalized from Sec.~\ref{Sec:cavityfeedback}, the measurement information can be collected from an arbitrary cavity of the coupled cavity array, and the measurement information from the $m$-th cavity reads
\begin{equation} \label{eq:Mcavitymeasure}
\begin{aligned}
I_c^{(m)}(t) & = \sqrt{2} \left \langle x_m \right\rangle + \frac{\xi_m(t)}{\sqrt{\eta}},
\end{aligned}
\end{equation}
where $x_m= \left(a_m+a^{\dag}_m\right)/\sqrt{2}$. Then we can derive the feedback dynamics according to the feedback operator similarly as in Section~\ref{Sec:OneCavityMF}.

For multiple coupled cavities, we take the feedback operator $G = \sum_m \left[ \beta_x^{(m)} x_m + \beta_p^{(m)} p_m\right]$ with the feedback strength $g_f$, the dynamics of the operators can be generalized by Eq.~(\ref{eq:xpdynamics}) after averaging the homodyne detection noise as
\begin{subequations} \label{eq:multicavityAsigma}
\begin{numcases}{}
\left\langle \dot{\sigma}_-^{(m)} \right\rangle = -ig_{m} e^{i\delta_m t}  \left\langle a_m \right\rangle - \mathcal{F}_m(t)\kappa^{(m)}\left\langle \sigma_-^{(m)} \right\rangle,\label{feedbacksigmaMm}\\
\left\langle \dot{a}_m \right\rangle =-i\Delta \left\langle a_m \right\rangle -i g_m e^{-i\delta_m t} \left\langle \sigma_-^{(m)} \right\rangle - \kappa_m  \left\langle a_m \right\rangle  - ig_f \left ( \beta_x^{(m)} + i\beta_p^{(m)}\right) \left\langle x_m \right\rangle -iJ_c \left(  \left\langle a_{m-1} \right\rangle + \left\langle a_{m+1} \right\rangle\right).
\end{numcases}
\end{subequations}
Obviously, the feedback drive will not directly influence the dynamics of $\left\langle \sigma_-^{(m)} \right\rangle$ according to Eq.~(\ref{feedbacksigmaMm}), but can influence the atomic dynamics via the atom-cavity couplings.

\begin{remark}
Upon Assumption~\ref{AssumeMultiplyJC}, when $g_m = 0$, $\left\langle \sigma_-^{(m)} \right\rangle$ can converge to zero, while $\left \langle a_m \right \rangle$ can be unstable if the feedback parameter $g_f$ is large enough compared with  $\left|\mathcal{F}_m(t)\right|$ and $\kappa$. 
\end{remark}

Denote $X_u = \left [ \mathbf{a}^R_1, \mathbf{a}^I_1,  \cdots, \mathbf{a}^R_M, \mathbf{a}^I_M\right]$, and $X_s = \left[ \mathbf{s}^R_1, \mathbf{s}^I_1, \cdots, \mathbf{s}^R_M, \mathbf{s}^I_M\right]$. Then Eq.~(\ref{eq:multicavityAsigma}) can be equivalently written as
\begin{equation} \label{con:StableUnstableSet}
\begin{aligned}
\begin{bmatrix}
\dot{X}_u\\
\dot{X}_s
\end{bmatrix}
&=\left( \begin{bmatrix}
 \mathbf{\mathcal{A}}_u & 0 \\
0 &\mathbf{\mathcal{A}}_s
\end{bmatrix}
+  \begin{bmatrix}
\mathcal{Q}_1(t) & \mathcal{Q}_2(t) \\
\mathcal{Q}_3(t) &\mathcal{Q}_4(t)
\end{bmatrix}\right)\begin{bmatrix}
X_u\\
X_s
\end{bmatrix},
\end{aligned}
\end{equation}
where
\begin{equation} \label{con:AsMcavity}
\begin{aligned}
\mathbf{\mathcal{A}}_u =&\begin{bmatrix}
\sqrt{2}g_f\beta_p^{(1)}-\kappa  &\Delta & \cdots & 0\\
\sqrt{2}g_f\beta_x^{(1)}-\Delta &-\kappa & \cdots &  0\\
\vdots & \vdots & \ddots  &\vdots\\
0 & 0 & \cdots   &\Delta\\
0 & 0 & \cdots  &-\kappa
\end{bmatrix}  +\begin{bmatrix}
0_{2\times2} &\mathbf{\mathcal{A}}_{J}  &0_{2\times2} &\cdots &0_{2\times2} &0_{2\times2}\\
\mathbf{\mathcal{A}}_{J} &0_{2\times2}  &\mathbf{\mathcal{A}}_{J} &\cdots &0_{2\times2} &0_{2\times2}\\
\vdots &\vdots &\vdots &\ddots &\vdots &\vdots \\
0_{2\times2} &0_{2\times2} &0_{2\times2} &\cdots & 0_{2\times2} &\mathbf{\mathcal{A}}_{J}\\
0_{2\times2} &0_{2\times2} &0_{2\times2} &\cdots &\mathbf{\mathcal{A}}_{J} & 0_{2\times2}
\end{bmatrix},
\end{aligned}
\end{equation}
\begin{equation} \label{con:AsMcavity2}
\begin{aligned}
\mathbf{\mathcal{A}}_s &=- \begin{bmatrix}
\bar{R}_f^{(1)}  & -\bar{I}_f^{(1)} & \cdots & 0 & 0\\
\bar{I}_f^{(1)} &\bar{R}_f^{(1)} & \cdots & 0  & 0\\
\vdots & \vdots & \ddots &\vdots &\vdots\\
0 & 0 & \cdots &\bar{R}_f^{(M)}  & -\bar{I}_f^{(M)}\\
0 & 0 & \cdots & \bar{I}_f^{(M)} &\bar{R}_f^{(M)}
\end{bmatrix},
\end{aligned}
\end{equation}
$\mathbf{\mathcal{A}}_{J}$ is given in Eq.~(\ref{con:S}), $\bar{R}_f^{(m)} = \lim_{t\rightarrow \infty} \mathcal{F}^R_m(t) $, $\bar{I}_f^{(m)} = \lim_{t\rightarrow \infty} \mathcal{F}^I_m(t) $, $m=1,2,\cdots,M$, $\mathcal{Q}_1(t) = 0$, $\mathcal{Q}_2(t) = \mathbf{R}(t)$,  $\mathcal{Q}_3(t) = \mathbf{G}(t)$,  $\mathcal{Q}_4(t) = \mathbf{F}(t)-\mathbf{\mathcal{A}}_s$, $\mathbf{R}(t)$, $\mathbf{G}(t)$ and $\mathbf{F}(t)$ are given by Eqs.~(\ref{con:R}), (\ref{con:G}), and (\ref{con:F}), respectively.

For a special case that $g_m = 0$, $P_2(t) = P_3(t) = 0$, the dynamics of $X_u$ and $X_s$ are decoupled in Eq.~(\ref{con:StableUnstableSet}). Obviously, $X_s$ will converge to zero if only the interaction between the atom and the environment becomes Markovian. However, $X_u$ can be unstable because of feedback controls. 

When $g_m \neq 0$, $X_u$ and $X_s$ are always coupled, the stability of Eq.~(\ref{con:StableUnstableSet}) can be determined by the following theorem based on the function $\mu(\cdot)$ given by Definition~\ref{defmu} in Appendix~\ref{sec:NonlinearAppendix}.
\begin{theorem} \label{StaUsta}
When the interaction between the atom and the environment becomes asymptotically Markovian in Eq.~(\ref{con:StableUnstableSet}), $X_u$  approaches zero when $t\rightarrow \infty$ provided that
\begin{equation}  \label{con:ProUScondition}
\int_{t_0}^{t_0+T} \mu \begin{Bmatrix}\begin{bmatrix}
\mathbf{\mathcal{A}}_u &\mathcal{Q}_2(\tau) \\
\mathcal{Q}_3(\tau) &\mathbf{\mathcal{A}}_s
\end{bmatrix}\end{Bmatrix} \mathrm{d}\tau< 0.
\end{equation}
\end{theorem}
\begin{proof}
When the interaction between the atom and the environment becomes asymptotically Markovian, $\lim_{t\rightarrow \infty} \mathcal{Q}_4(t) = 0$ in Eq.~(\ref{con:StableUnstableSet}) with $\mathcal{Q}_1(t) = 0$, and the matrix $\begin{bmatrix}
\mathbf{\mathcal{A}}_u & \mathcal{Q}_2(t) \\
\mathcal{Q}_3(t) &\mathbf{\mathcal{A}}_s
\end{bmatrix}$ is periodic. Then the result holds according to the proof of Corollary~\ref{TimevaryPeriodCombine}.
\end{proof}

\begin{remark} \label{subspaceRemark}
As in Ref.~\cite{vrabel2019note}, a choice of the norm in Theorem~\ref{StaUsta} can be
\begin{equation} 
\begin{aligned} \label{con:MPMtr}
&\frac{1}{2}\lambda_{\max} \left( \begin{bmatrix}
\mathbf{\mathcal{A}}_u +\mathbf{\mathcal{A}}_u^{\rm T}  & \mathcal{Q}_2(t) + \mathcal{Q}_3^{\rm T}(t) \\
\mathcal{Q}_3(t) + \mathcal{Q}_2^{\rm T}(t)&\mathbf{\mathcal{A}}_s + \mathbf{\mathcal{A}}_s^{\rm T}
\end{bmatrix}\right) =\frac{1}{2}\lambda_{\max} \left( \begin{bmatrix} 
\bar{\mathbf{a}}_1 & \cdots & 0 & 0 & \cdots & 0\\
\vdots &\ddots & \vdots &\vdots &\ddots & \vdots\\
0 & \cdots & \bar{\mathbf{a}}_M & 0 & \cdots & 0\\
0 & \cdots & 0 &\bar{\mathbf{b}}_1 & \cdots & 0 \\
\vdots &\ddots & \vdots &\vdots &\ddots & \vdots\\
0 & \cdots & 0 &0 & \cdots & \bar{\mathbf{b}}_M
\end{bmatrix} 
\right),
\end{aligned}
\end{equation} 
where $\bar{\mathbf{a}}_m = \begin{bmatrix}
2\sqrt{2}g_f\beta_p^{(m)}-2\kappa_m &\sqrt{2}g_f\beta_x^{(m)} \\
\sqrt{2}g_f\beta_x^{(m)} & -2\kappa_m
\end{bmatrix}$, $\bar{\mathbf{b}}_m =
\begin{bmatrix}
-2\bar{R}_f^{(m)} & 0\\
0 &-2\bar{R}_f^{(m)}
\end{bmatrix}$. When $g_f$ and $\beta_x^{(m)}$ are small, the inequality in (\ref{con:ProUScondition}) can be easier to meet, then both $X_u$ and $X_s$ can converge to zero. An extreme case is that $\beta_x^{(m)} = 0$ and $\kappa_m > \sqrt{2}g_f \beta_p^{(m)}$, then the matrix in Eq.~(\ref{con:MPMtr}) is a diagonal matrix with all the diagonal elements being negative.
\end{remark}
\begin{figure}[h]
\centerline{\includegraphics[width=0.7\columnwidth]{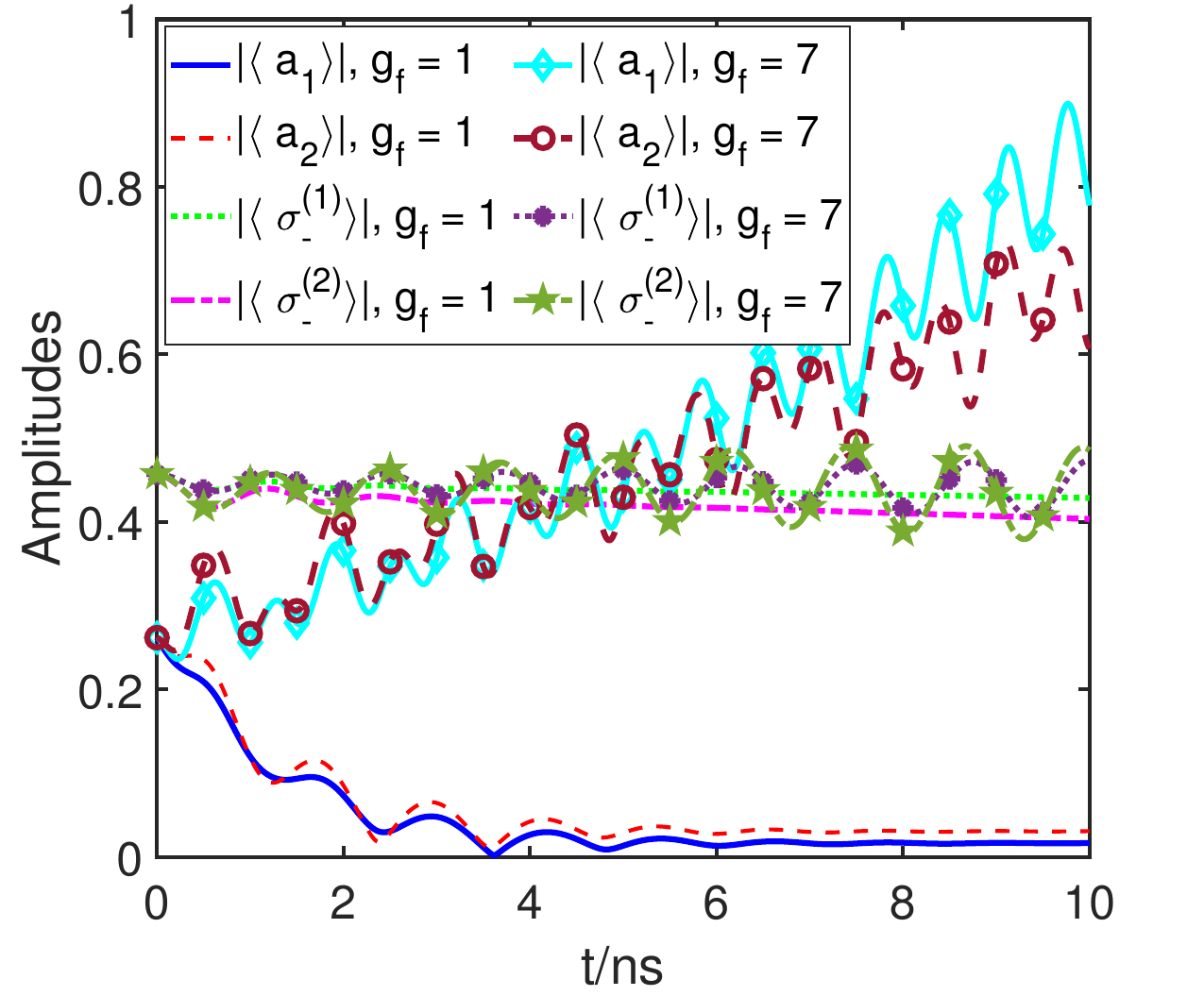}}
\caption{Measurement feedback control based on two coupled cavities.}
	\label{fig:TwocavityFB}
\end{figure}
As compared in Fig.~\ref{fig:TwocavityFB}, there are two coupled cavities, $J_c = 0.1$GHz, $\delta_1 =  \delta_2= 0.2$GHz, $\kappa = 1$GHz, $g_1 = 0.2$GHz, $g_2 = 0.4$GHz, $\Delta = 5$GHz, the two atoms in the two cavities are coupled identically with the non-Markovian environment  with $\omega_a = 43.98$GHz, and $\Omega = 38.98$GHz, $\chi  = 1$, $\kappa^{(1)}=\kappa^{(2)}=0.04$GHz, $\gamma = 2$GHz for the non-Markovian environment~\cite{NatureExperiment}. In the measurement feedback design, $\beta_x^{(1)} = \beta_x^{(2)} = 0$GHz, $\beta_p^{(1)} = \beta_p^{(2)} = 0.2$GHz. On one hand, when $g_f = 1$, $\sqrt{2}g_f \beta_p^{(1)} - \kappa = -0.7172 {\rm GHz} <0{\rm GHz}$. On the other hand, when $g_f = 7$, $\sqrt{2}g_f \beta_p^{(1)} -\kappa = 0.9799 {\rm GHz}>0{\rm GHz}$.  The two group of simulations agree with Remark~\ref{subspaceRemark}, showing that the stable and unstable subspaces can be modulated by tuning the feedback parameters.

The above analysis on multiple Jaynes-Cummings model realized by cascade cavity networks can be a potential approach for realizing non-Markovian interactions, according to the introduction on the frequency spectrum properties in~\cite{lombardo2014photon}. Besides, the coupled cavity array can be further used to realize waveguide quantum electrodynamics~\cite{zhou2008controllable}, constructing various complex non-Markovian quantum networks~\cite{ding2023quantum,ding2023quantumNlevel,ding2023Measurement}. This generalization can introduce more challenges on quantum control due to the fact that the non-Markovian environment can destruct the purity of quantum states~\cite{Hannes2015Quantum}, which can further influence the control properties such as the creation of entangled quantum states~\cite{zhang2010protecting} and photons~\cite{ding2023Measurement}.

\section{Conclusion}\label{Sec:conclusion}
In this paper, we study quantum control dynamics based on a cavity-QED system, focusing on the non-Markovian interactions between the quantum system and its environment. 
The evolution of parameters representing the atom's non-Markovian decay into the environment is described by nonlinear dynamics.
The transition to Markovian interactions between the quantum system and the environment is explained by the stability of these nonlinear processes. Consequently, the dynamics of the multi-level system in a non-Markovian environment can be modeled by linear time-varying equations. 
Manipulation of atomic states and photons in the cavity is then achieved using open-loop and closed-loop control methods that utilize quantum measurement feedback.
This approach can be extended to scenarios involving multiple coupled cavities described by high-dimensional linear time-varying equations, where feedback control can further influence the dynamics between stable and unstable subspaces. 
\begin{appendices}

\section{Derivation of the non-Markovian parameter dynamics} \label{sec:NonMarkovParemeters}
We take the atom and cavity as a whole represented with the state $|\psi(t)\rangle$, then the combined system interacts with the bath via the multi-level atom. The interaction between the atom and environment can be represented with the Hamiltonian
\[
H_{\rm AE} = \sum_{n=1}^{N-1}\sum_{\omega} \chi_{\omega}^{(n)} \left (|n+1\rangle \langle n|b_{\omega} + |n\rangle \langle n+1| b_{\omega}^{\dag} \right ),
\]
and the detailed meaning is introduced after Eq.~(\ref{con:TwoNlevelatomHam}) in the main text. Here we take a representative simplified case that the quantum system is coupled to the environment via operators $L_n$. 
The evolution of quantum states is influenced by its stochastic interaction with the environment as~\cite{diosi1998non,diosi1997nonPLA,yu1999non,de2005two}
\begin{equation} \label{con:NonMarkovStep1}
\frac{\delta |\psi(t)\rangle}{\delta z_s} = \sum_{n=1}^{N-1} f_n(t,s)L_n|\psi(t)\rangle,
\end{equation}
where $z_s$ is a complex-valued Wiener process at time $s$, $f_n(t,s)$ is a two-time function to be determined, related to the two-time points  $t$ and $s$~\cite{diosi1998non}. Combined with Eq.~(\ref{con:Loperator1}),
the differential of Eq.~(\ref{con:NonMarkovStep1}) upon the time $t$ reads
\begin{equation} \label{con:NonMarkovStep2}
\frac{\delta |\dot{\psi}\rangle}{\delta z_s}=\sum_{n=1}^{N-1} \left[ \frac{\partial f_n(t,s)}{\partial t}L_n|\psi\rangle + f_n(t,s) L_n |\dot{\psi}\rangle\right],
\end{equation}
where $|\dot{\psi}\rangle$ represents the derivation according to time $t$, and 
\begin{equation} \label{con:NonMarkovStep3}
\begin{aligned}
\frac{\partial f_n(t,s)}{\partial t}L_n|\psi\rangle &=\frac{\delta |\dot{\psi}\rangle}{\delta z_s} - f_n(t,s) L_n |\dot{\psi}\rangle\\
&=-iH\frac{\delta |\psi\rangle}{\delta z_s} + i f_n(t,s)L_n H |\psi\rangle + L_n\frac{\delta |\psi\rangle}{\delta z_s} z_t   -\frac{\delta }{\delta z_s} L_n^{\dag}\int_0^t \alpha(t,s)\frac{\delta |\psi\rangle}{\delta z_s}\mathrm{d}s -f_n(t,s) L_nL_n |\psi\rangle z_t \\
&~~~ + f_n(t,s) L_n  L_n^{\dag}\int_0^t \alpha(t,s)\frac{\delta |\psi\rangle}{\delta z_s}\mathrm{d}s\\
&=i\left[L_n,H\right]f_n(t,s)|\psi\rangle +F_n(t) \left ( L_nL_n^{\dag} - L_n^{\dag}L_n\right ) f_n(t,s)L_n|\psi\rangle,
\end{aligned}
\end{equation}
with the time-dependent function $F_n(t)$ defined by Eq.~(\ref{con:NonMarkovStep3Ft})  and $\alpha(t,s)$ given by Eq.~(\ref{con:NonMarkovAlpha}) in the main text.

Then combined with Eq.~(\ref{con:Loperator1}),  for a specific $L_n$, Eq.~(\ref{con:NonMarkovStep3}) can be rewritten as  $\partial_t f_n(t,s) L_n|\psi\rangle 
=i\omega_n f_n(t,s)L_n|\psi\rangle + F_n(t)  \kappa_n f_n(t,s)L_n|\psi\rangle$, further combined with Eq.~(\ref{con:NonMarkovStep3Ft}), we have
\begin{equation} \label{con:FtEquation}
\begin{aligned}
\dot{F}_n(t)&= \frac{\gamma \chi_n}{2} - (\gamma + i\Omega)F_n(t) + \int_0^t \alpha(t,s) \left [i\omega_n f_n(t,s) + F_n(t) \kappa_nf_n(t,s) \right] \mathrm{d}s \\
&= \kappa_n F_n^2(t)  - \left(\gamma + i\Omega - i\omega_n\right)F_n(t) + \frac{\gamma \chi_n}{2},
\end{aligned}
\end{equation}
where $f_n(t,t) = \chi_n$~\cite{diosi1998non}.

\section{Derivation of the non-Markovian master equation} \label{sec:NonMarkovMaster}
The non-Markovian master equation has been introduced in Refs.~\cite{diosi1998non,yang2012nonadiabatic}. In this Appendix, we briefly introduce the derivation of the non-Markovian control equation~(\ref{con:SSENonMarUpdate0}) in the main text. Based on $f_n(t,s)$ clarified in Appendix~\ref{sec:NonMarkovParemeters}, combine Eq.~(\ref{con:SSENonMar}) with Eq.~(\ref{con:NonMarkovStep1}) after the normalization, we have the following equation for the state vector,
\begin{equation} \label{con:SSENonMarUpdateNorm}
\begin{aligned}
&\frac{\mathrm{d}}{\mathrm{d}t} |\psi(t)\rangle = -iH |\psi(t)\rangle +\sum_{n=1}^{N-1} \left (L_n - \left\langle L_n\right\rangle \right) |\psi(t)\rangle z_t - \sum_{n=1}^{N-1}\left (L_n^{\dag} - \left\langle L_n^{\dag}\right\rangle \right)\int_0^t \alpha(t,s) f_n(t,s)L_n|\psi(t)\rangle\mathrm{d}s.
\end{aligned}
\end{equation}

Considering that $\rho(t) = E_z(|\psi(t)\rangle \langle \psi(t) |)$, and denote $P_t\left(z_t^*\right) =  |\psi(t)\rangle \langle \psi(t) |$, then $E_z \left (P_tz_t \right) =\sum_{n=1}^{N-1} \int_0^t E_z\left[ z_tz_s^*\right]  f_n^*(t,s)\rho(t)L_n^{\dag} \mathrm{d}s$ and $E_z \left (P_tz_t^* \right) =\sum_{n=1}^{N-1}\int_0^t  E_z\left[ z_t^*z_s\right] f_n(t,s)L_n\rho(t)\mathrm{d}s$~\cite{de2005two}.
Thus the dynamics of the density matrix reads
\begin{equation} \label{con:SSENonMarUpdateAppendix}
\begin{aligned}
\dot{\rho}=&-i[H,\rho]+\sum_{n=1}^{N-1}\left \{\int_0^t \mathrm{d}s E_z\left[ z_tz_s^*\right]  f_n^*(t,s) + \int_0^t \mathrm{d}s E_z\left[ z_t^*z_s\right] f_n(t,s)L_n  \rho L_n^{\dag} \right.\\
&\left.-  \int_0^t \alpha(t,s) f_n(t,s) \mathrm{d}s  L_n^{\dag}L_n\rho- \rho L_n^{\dag}L_n \int_0^t \alpha^*(t,s) f_n^*(t,s)\mathrm{d}s\right\} ,
\end{aligned}
\end{equation}
and can be further written as Eq.~(\ref{con:SSENonMarUpdate0}) in the main text~\cite{yang2012nonadiabatic}.

\section{Some results on nonlinear and linear dynamics} \label{sec:NonlinearAppendix}
In this appendix, we recall some concepts and conclusions on nonlinear and linear time-varying systems for the clarifications in  Sec.~\ref{Sec:ModelNonMark} and Sec.~\ref{Sec:quantumcontrolOpen}. 
\begin{myDef}[\cite{SIAMBIBO,lin1967bounded}] \label{defstable}
The system in Eq.~(\ref{con:NonlinearXFControl}) is said to be bounded-input bounded-output (BIBO) stable if for every $a_u\geq 0$ and $\alpha_X \geq 0$ there is a finite number $\beta_X = \beta_X\left(a_u,\alpha_X\right)$ such that for every initial condition $X_F\left (t_0 \right)$ with $X_F\left (t_0 \right )\leq \alpha_X$ and every sequence $u$ with $\|u\| \leq a_u$,
\[
X_F\left (u,t';X_F\left (t_0 \right ),t_0 \right ) \leq \beta_X,
\]
for $t'>t_0$.
\end{myDef}

\begin{myDef} [\cite{lasalle1960some}]\label{invariantset}
For a state vector $X_F(t)$ evolves in the time domain as $\dot{X}_F(t) = \mathcal{F}\left( X_F\right)$,  a set $S$ is said to be forward invariant if $X_F\left(t_0\right) \in S$, $X_F\left(t\right) \in S$ for any $t>t_0$.
\end{myDef}

\begin{mypro}[\cite{SIAMBIBO}]\label{BIBOProposition}
~Suppose for each $u\geq 0$, there exists a positive Lyapunov function $V_u$ for the system with the input $u$ and $\dot{V}_u(t,X_f) \leq 0$.
When the input satisfies $\|u\| \leq a_u$, the system is BIBO stable.
\end{mypro}

\begin{mypro}[\cite{chen2004stability}] \label{Chenstable}
 For the nonlinear system
\[
\dot{\textbf{X}}(t) = \textbf{f}(\textbf{X}(t),t) + \textbf{h}(\textbf{X}(t),t), 
\]
where $\textbf{f}$ is continuously differentiable with $\textbf{f}(0,t) = 0$ and $\textbf{h}(\textbf{X}(t),t)$ is a persistent perturbation: if the system is uniformly and asymptotically stable about its equilibrium $\textbf{X}^* =0$ when $\textbf{h}(\textbf{X}(t),t) = 0$, and there are two positive constants $\tilde{\delta}_1$ and $\tilde{\delta}_2$ such that $\|\textbf{h}(\textbf{X}(t),t)\|<\tilde{\delta}_1$ for $t\in[0,\infty)$ and $\|\textbf{h}(\textbf{X}(0),0)\|<\tilde{\delta}_2$, then the persistently perturbed system remains to be stable in the sense of Lyapunov. 
\end{mypro}

\begin{myDef}[\cite{vrabel2019note}]\label{UES}
A linear time varying system is uniformly exponential stable (UES) if and only if there exist positive constants $K$ and $\tilde{\alpha}$ such that
$\|\Phi\left(t,\tau\right) \| \leq K e^{-\tilde{\alpha}\left( t- \tau\right)}$,  
for $t_0 \leq \tau < t < \infty$.
\end{myDef}

\begin{myDef}[\cite{vrabel2019note}]\label{defmu}
For a real-valued continuous matrix $A(t)\in \textbf{R}^{n\times n}$, the logarithmic norm can be defined when $t\geq 0$ as
\[
\mu[A(t)] = \lim_{h\rightarrow 0^+} \frac{\| I_n + h A(t)  \| - 1}{h},
\]
and
$ \Pi^+(t) \triangleq \int_{t_0}^t \mu [A(\tau)] \mathrm{d}\tau$.
\end{myDef}

\section{Proof of Corollary~\ref{TimevaryPeriodCombine}}\label{sec:linearAppendix}
\begin{proof}
The proof is based on two separated parts as discussed in  Refs.~\cite{vrabel2019note,bellman2008stability}. According to Ref.~\cite{vrabel2019note}, when (a) is satisfied, the linear time-varying system $\dot{\hat{\mathbb{X}}} =\bar{\mathbb{A}}_n(t) \hat{\mathbb{X}}$ is uniformly exponential stable according to Definition~\ref{UES}, and $\lim_{t\rightarrow \infty } \hat{\mathbb{X}} = 0$. Based on this, when (b) is satisfied, the state vector in Eq.~(\ref{con:TimevariedEquation2}) approaches zero, as shown in Ref.~\cite{bellman2008stability} (Chapter 2, Theorem 4).  
\end{proof}

\end{appendices}
\section*{Acknowledgements}
This work is supported by the ANR project “Estimation et controle des systèmes quantiques ouverts” Q-COAST Project ANR- 19-CE48-0003 and the ANR project IGNITION ANR-21-CE47-0015, Quantum Science and Technology-National Science and Technology Major Project 2023ZD0300600, Guangdong Provincial Quantum Science Strategic Initiative No. GDZX2303007, and Hong Kong Research Grant Council (RGC) under Grant No. 15213924. 
\bibliographystyle{IEEEtran}
\bibliography{MFQEC}

\end{document}